\documentclass{article}

\usepackage[left=1.25in,top=1.25in,right=1.25in,bottom=1.25in,head=1.25in]{geometry}
\usepackage{amsfonts,amsmath,amssymb,amsthm}
\usepackage{verbatim,float,color,dsfont}
\usepackage{graphicx,subfigure,url}
\usepackage{natbib}

\newtheorem{algorithm}{Algorithm}

\newtheorem{lemma}{Lemma}

\newcommand{\argmin}{\mathop{\mathrm{arg\,min}}}

\newcommand{\st}{\mathop{\mathrm{subject\,\,to}}}

\def\R{\mathbb{R}}

\def\sign{\mathrm{sign}}

\def\col{\mathrm{col}}
\def\row{\mathrm{row}}
\def\nul{\mathrm{null}}
\def\rank{\mathrm{rank}}
\def\nuli{\mathrm{nullity}}
\def\half{\frac{1}{2}}
\def\hbeta{\hat\beta}
\def\hu{\hat{u}}

\def\ty{\tilde{y}}
\def\tA{\widetilde{A}}
\def\tD{\widetilde{D}}
\def\tG{\widetilde{G}}
\def\tP{\widetilde{P}}
\def\tQ{\widetilde{Q}}
\def\tR{\widetilde{R}}

\def\cB{\mathcal{B}}

\def\cG{\mathcal{G}}

\def\wBox{{\color{white} \Box}}
\def\sw{\setlength{\arraycolsep}{1pt}}
\def\sh{\renewcommand{\arraystretch}{0.75}}

\title{Efficient Implementations of the Generalized
Lasso Dual Path Algorithm}
\author{
\parbox{2in}{\centering
Taylor B. Arnold \\
AT\&T Labs Research}
\hspace{0.5in}
\parbox{2in}{\centering
Ryan J. Tibshirani \\
Carnegie Mellon University}
}
\date{}

\begin{document}
\maketitle

\begin{abstract}
We consider efficient implementations of the generalized lasso dual
path algorithm of \citet{genlasso}.  We first describe a generic
approach that covers any penalty matrix $D$ and any (full column rank)
matrix $X$ of predictor variables.  We then describe fast
implementations for the special cases of
trend filtering problems, fused lasso problems, and sparse fused lasso
problems, both with $X=I$ and a general matrix $X$.  These specialized
implementations offer a considerable improvement over the generic
implementation, both in terms of numerical stability and efficiency of
the solution path computation.  These algorithms are all available for
use in the {\tt genlasso} R package, which can be found in the CRAN
repository.

\end{abstract}

\section{Introduction}
\label{sec:intro}

In this article, we study computation in the generalized lasso problem
\citep{genlasso}
\begin{equation}
\label{eq:genlasso}
\hbeta = \argmin_{\beta \in \R^p} \,
\half \|y-X\beta\|_2^2 + \lambda\|D\beta\|_1,
\end{equation}
where $y\in\R^n$ is an outcome vector, $X\in\R^{n\times p}$ is a
predictor matrix, $D\in\R^{m\times p}$ is a penalty
matrix, and $\lambda \geq 0$ is a regularization parameter.
The term ``generalized'' refers to the fact that problem
\eqref{eq:genlasso} reduces to the standard lasso problem
\citep{lasso,bp} when $D=I$, but yields different problems with
different choices of the penalty matrix $D$.  We will assume that $X$
has full column rank (i.e., $\rank(X)=p$), so as to ensure a unique
solution in \eqref{eq:genlasso} for all values of $\lambda$.

Our main contribution is to derive efficient implementations
of the generalized lasso dual path algorithm of \citet{genlasso}.
This algorithm computes the
solution \smash{$\hbeta(\lambda)$} in \eqref{eq:genlasso} over the
full range of regularization parameter values $\lambda \in
[0,\infty)$.
We present an efficient implementation for a general
penalty matrix $D$, as well as specialized, extra-efficient
implementations for two special classes of generalized lasso problems:
{\it fused lasso} and {\it trend filtering} problems.
The algorithms that we describe in this work
are all implemented in the {\tt genlasso} R package, freely
available on the CRAN repository \citep{cran}.

We note that the fused lasso and trend filtering problems are known,
well-established problems (early references for fused lasso are
\citet{fusion}, \citet{fuse}, and early works on trend filtering are
\citet{hightv}, \citet{l1tf}).  These problems are not original to the
generalized lasso framework, but the latter framework simply provides
a useful,  unifying perspective from which we can study them.  We give
a brief overview here; see the aforementioned references for more
discussion, or Section 2 of \citet{genlasso}, and also Section
\ref{sec:crime} of this paper, for examples and figures. 

In the first problem class, the fused lasso setting, we think of the
components of $\beta \in \R^p$ as corresponding to the nodes of a
given undirected graph $G$, with edge set $E \subseteq \{1,\ldots
p\}^2$. If $E$ has $m$ edges, enumerated $e_1,\ldots e_m$, then the
fused lasso penalty matrix $D$ is $m \times p$, with one row
for each edge in $E$.  In particular, if $e_\ell = (i,j)$, then the
$\ell$th row of $D$ is
\begin{equation*}
D_\ell = (0, \ldots \underset{\substack{\;\;\uparrow \\ \;\;i}}{-1},
\ldots \underset{\substack{\uparrow \\ j}}{1}, \ldots 0) \in \R^p,
\end{equation*}
i.e., $D_\ell$ contains all zeros except for a $-1$ and $1$ in the
$i$th and $j$th components (equivalent to this would be a $1$ and $-1$
in the $i$th and $j$th components). The fused lasso penalty term is
hence
\begin{equation*}
\|D\beta\|_1 = \sum_{e_\ell = (i,j) \in E} |\beta_i-\beta_j|.
\end{equation*}
The effect of this penalty in \eqref{eq:genlasso} is that many of the
terms \smash{$|\hbeta_i-\hbeta_j|$} will be zero at the solution
\smash{$\hbeta$}; in other words, the solution exhibits a piecewise
constant structure over connected subgraphs of $G$.
To relate the fused lasso penalty matrix to concepts from graph
theory, note that $D$ as described above is the
{\it oriented incidence matrix} of the undirected graph
$G$.
% \footnote{Recall that the oriented
% incidence matrix of an undirected graph $G$ is the incidence matrix
% of $G$ after we assign arbitrary orientations to its edges.  Hence 
% it is somewhat of an abuse of terminology to refer ``the'' oriented
% incidence matrix of $G$, as it is nonunique; but in our case, as in
% many other cases, the orientations do not matter.} 
This means that $L=D^T D$ is the {\it Laplacian matrix} of $G$---a
realization that will be useful for our work in Section
\ref{sec:fusedlasso}. 

An important special case to mention is the {\it 1-dimensional fused
lasso}, in which the components of $\beta$ correspond to successive
positions on a 1-dimensional grid, so $G$ is the chain graph
(with edges $(i,i+1)$, $i=1,\ldots p-1$), and the penalty matrix is
\begin{equation}
\label{eq:d1d}
D = \left[\begin{array}{rrrrrr}
-1 & 1 & 0 & \ldots & 0 & 0 \\
0 & -1 & 1 & \ldots & 0 & 0 \\
\vdots & & & & & \\
0 & 0 & 0 & \ldots & -1 & 1
\end{array}\right].
\end{equation}
The solution \smash{$\hbeta$} is therefore piecewise constant across
the underlying positions.  (A clarifying note: the original work of
\citet{fusion}, \citet{fuse} considered only this 1-dimensional setup,
and the generalization to a graph setting came in later works. Also,
\citet{fuse} defined the fused lasso criterion with an additional
$\ell_1$ penalty on coefficients $\beta$ themselves; we now refer to
this as the {\it sparse fused lasso} problem.  It is not fundamentally
different from the version we consider here, with pure fusion, and can 
be handled by the described path algorithm; see Section
\ref{sec:lapsfl}.) 

The second problem class, trend filtering, also starts with the
assumption that the components of $\beta$
correspond to positions on an underlying 1d grid,
like the 1d fused lasso; but trend filtering more generally produces
a solution \smash{$\hbeta$} that bears the structure of a piecewise
$k$th degree polynomial function over the underlying positions, where
$k \geq 0$ is a given integer.  To accomplish this, the trend
filtering penalty matrix is taken to be $D=D^{(k+1)}$, the
$(p-k-1)\times p$ discrete difference operator of order $k+1$.
These discrete derivative operators can be defined recursively, by
letting $D^{(1)}$ be the $(p-1)\times p$ first difference matrix in
\eqref{eq:d1d}, and
\begin{equation*}
D^{(k+1)} = D^{(1)} D^{(k)} \;\;\; \text{for}\;\, k=1,2,3,\ldots.
\end{equation*}
(In the above, $D^{(1)}$ is the $(p-k-1)\times (p-k)$ version of the
first difference matrix.)  Note that the 1d fused lasso is exactly
a special case of trend filtering with $k=0$.
For a general order $k \geq 0$, the matrix $D^{(k+1)}$ is
banded with bandwidth $k+2$, and a straightforward calculation shows
that the penalty term in \eqref{eq:genlasso} can be written explicitly
as
\begin{equation*}
\|D^{(k+1)}\beta\|_1 = \sum_{i=1}^{p-k-1} \left|
\sum_{j=i}^{i+k+1} (-1)^{j-i} {k+1 \choose j-i} \beta_j \right|.
\end{equation*}
We refer the reader to \citet{trendfilter} for a study of trend
filtering as a signal estimation method (i.e., for $X=I$), where it is
shown to have desirable statistical properties in a nonparametric  
regression context.
 
In what follows, we describe the dual path algorithm of
\citet{genlasso}, and then begin discussing strategies for its
implementation. First, though, we briefly review other computational
approaches for the generalized lasso problem \eqref{eq:genlasso}.

\subsection{Related work}
\label{sec:related}

There are many algorithms, aside from the dual path algorithm
central to this paper, for solving the convex problem
\eqref{eq:genlasso} and its various special cases.  It will be helpful
to distinguish between algorithms that solve \eqref{eq:genlasso} at
fixed values of the tuning parameter $\lambda$, and algorithms 
that sweep out the entire path of solutions as a continuous function
of $\lambda$.  

In the former case, when the solution is
desired a fixed value of $\lambda$, a number of
more or less standard convex optimization techniques can be applied.
For arbitrary matrices $X,D$, problem \eqref{eq:genlasso} can be
recast as a quadratic program, so that, e.g., we may use the standard 
interior point methods common to quadratic  and conic programming
problems.  We can also use the alternating direction method of
multipliers (ADMM) for general $X,D$.  For certain instantiations of
$X,D$, there are 
faster, more specialized techniques.  For example, when $X=I$ and $D$
is the 1d fused lasso matrix, problem \eqref{eq:genlasso} can be
solved in linear time via a taut string method \citep{tautstring},
or dynamic programming \citep{fusedp}.  When $X=I$ and $D$ is the
fused lasso matrix over a graph, a clever parametric max flow
approach \citep{chambolle09} applies. When $X=I$ and $D$ is the trend  
filtering matrix, highly efficient and specialized interior 
point methods \citep{l1tf} or ADMM algorithms \citep{fasttf} are
available. Finally, when $X$ is an arbitrary matrix and $D$
falls into any one of the above categories, one can implement a
proximal gradient algorithm, with each proximal evaluation
utilizing one of the specialized techniques just described.

In terms of path algorithms, the literature is more sparse.  
For the lasso problem, the well-known least angle regression algorithm
of \citet{lars} computes the full solution path (see also
\citet{homotopy1,homotopy2}).  For fused lasso problems,
\citet{hoefling} describes a path algorithm based on max flow
subroutines, which efficiently tracks the path in the 
direction opposite to the one we consider (i.e., starts with
$\lambda=0$ and ends at $\lambda=\infty$).  
For the generalized lasso problem, \citet{zhou2013}
propose a path algorithm from the primal perspective; however, their
work assumes $D$ to have full row rank, which does not hold in
many cases of interest (such as the fused lasso over a graph with more
edges than nodes).  
The dual path algorithm of \citet{genlasso} has the
advantage that it 
operates in a single, unified framework that allows $D$ to be
completely general, but is also flexible enough to permit efficient
specialized versions when $D$ takes specific forms.  Given the 
magnitude of related work, we do not give detailed comparisons to
alternative methods, but instead focus on fast, stable
implementations of the generalized lasso dual path algorithm.

\subsection{The dual path algorithm}
\label{sec:dualpath}

We recall the details of the dual path algorithm for the
generalized lasso problem. We do not place any
assumptions on $D \in
\R^{m\times n}$, but we do assume that $X$ has full column rank, which
implies a unique solution in \eqref{eq:genlasso} for all $\lambda$.
As its name suggests, the dual path algorithm actually computes a
solution path of the equivalent {\it dual problem} of
\eqref{eq:genlasso}, instead of solving \eqref{eq:genlasso} directly.

\subsubsection{The signal approximator case, $X=I$}

It helps to first consider the
``signal approximator'' case, $X=I$.  In this case, for any
fixed value of $\lambda$, the dual of problem \eqref{eq:genlasso} is:
\begin{equation}
\label{eq:genlassodual1}
\hu \in \argmin_{u\in\R^m} \,
\half \|y-D^T u\|_2^2
\;\;\st\;\, \|u\|_\infty \leq \lambda,
\end{equation}
and the primal and dual solutions, \smash{$\hbeta$} and $\hu$, are
related by:
\begin{equation}
\label{eq:primaldual1}
\hbeta = y-D^T \hu.
\end{equation}
We note that, though the primal solution is unique, the dual
solution need not be unique (this is reflected by the element
notation in \eqref{eq:genlassodual1}).   

The path algorithm proposed
\citet{genlasso} computes a solution path $\hu(\lambda)$ of the dual
problem, beginning at $\lambda=\infty$ and progressing down to
$\lambda=0$; this gives the primal solution path
\smash{$\hbeta(\lambda)$}
using the transformation in \eqref{eq:primaldual1}.   At a high level,
the algorithm keeps track of the coordinates of the
computed dual solution $\hu(\lambda)$ that are equal to $\pm \lambda$,
i.e., that lie on the boundary of the constraint region
$[-\lambda,\lambda]^m$, and it determines critical
values of the regularization parameter, $\lambda_1 \geq \lambda_2
\geq \ldots$, at which coordinates of this solution hit or
leave the boundary.
We outline the algorithm below; in terms of notation, we
write $D_S$ to extract the rows of $D$ in $S \subseteq \{1,\ldots
m\}$, and we use $D_{-S}$ as shorthand for $D_{\{1,\ldots m\}
  \setminus S}$.

\begin{algorithm}[\textbf{Dual path algorithm for the generalized
    lasso, $X=I$}]
\label{alg:dualpath1}
\hfill\par
\smallskip
\smallskip
Given $y \in \R^n$ and $D \in \R^{m\times n}$.
\begin{enumerate}
\item Compute $\hu$, the minimum $\ell_2$ norm solution
of
\begin{equation*}
\min_{u\in\R^m} \, \|y-D^T u\|_2^2.
\end{equation*}

\item Compute the first hitting time $\lambda_1$, and the hitting
coordinate $i_1$.  Record the solution $\hu(\lambda)=\hu$ for
$\lambda \in [\lambda_1, \infty)$.  Initialize $\cB=\{i_1\}$,
$s=\sign(\hu_{i_1})$, and $k=1$.

\item While $\lambda_k>0$:
\begin{enumerate}
\item Compute $\hat{a}$ and \smash{$\hat{b}$}, the minimum $\ell_2$
norm solutions of
\begin{equation*}
\min_{a \in \R^{m-|\cB|}} \, \|y-D_{-\cB}^T a \|_2^2 \;\;\;\text{and}\;\;\;
\min_{b \in \R^{m-|\cB|}} \, \|D_\cB^T s - D_{-\cB}^T b \|_2^2,
\end{equation*}
respectively.

\item Compute the next hitting time and the next leaving time.  Let
$\lambda_{k+1}$ denote the larger of the two; if the hitting time is
larger, then add the hitting coordinate to $\cB$ and its sign to $s$,
otherwise remove the leaving coordinate from $\cB$ and its sign from
$s$.  Record the solution
\smash{$\hu(\lambda) = \hat{a}-\lambda \hat{b}$} for
$\lambda \in [\lambda_{k+1},\lambda_k]$, and update $k=k+1$.
\end{enumerate}
\end{enumerate}
\end{algorithm}

The main computational effort lies in Steps 1 and 3(a). In words:
starting with the set $\cB=\emptyset$, we repeatedly solve least
squares problems of the form $\min_x \|c-D_{-\cB}^T x\|_2^2$---which
is the same as solving linear systems $D_{-\cB}D_{-\cB}^T x = D_{-\cB}
c$---as elements are added to or deleted from $\cB$, that is,
$D_{-\cB}$ either loses or gains one row.
A caveat is that we always require the minimum $\ell_2$ norm
solution (but this is only an important distinction when the
solution is not unique).  Steps 2 and 3(b) are computationally
straightforward, as they utilize the results of Steps 1 or 3(a) in
a simple way; see Section 5 of \citet{genlasso} for specific details.

\subsubsection{The general $X$ case}

For a general $X$, with $\rank(X)=p$, the dual problem of
\eqref{eq:genlasso} can be written as:
\begin{equation}
\label{eq:genlassodual2}
\hu \in \argmin_{u\in\R^m} \,
\half \|XX^+y-(X^+)^T D^T u\|_2^2
\;\;\st\;\, \|u\|_\infty \leq \lambda,
\end{equation}
where $X^+ \in \R^{p\times n}$ denotes the Moore-Penrose pseudoinverse
of $X \in \R^{n\times p}$ (recall that for rectangular $X$, we take
$X^+=(X^T X)^+ X^T$), and the primal and dual solutions are related
by:
\begin{equation}
\label{eq:primaldual2}
X\hbeta = XX^+y- (X^+)^T D^T\hu.
\end{equation}
Though it may initially look more complicated,
the dual problem \eqref{eq:genlassodual2} is of the exact same form as
\eqref{eq:genlassodual1}, the dual in the signal approximator case,
but with a different outcome $\ty=XX^+ y$
and penalty matrix \smash{$\tD = DX^+$}.  Hence, modulo a
transformation of inputs, the same algorithm can be applied.

\begin{algorithm}[\textbf{Dual path algorithm for the generalized
    lasso, general $X$}]
\label{alg:dualpath2}
\hfill\par
\smallskip
\smallskip
Given $y \in \R^n$,
$D \in \R^{m\times p}$, and $X \in \R^{n\times  p}$ with
$\rank(X)=p$.
\begin{enumerate}
\item Compute $\ty=XX^+ y \in \R^n$ and
\smash{$\tD=DX^+ \in \R^{m\times n}$}.
\item Run Algorithm \ref{alg:dualpath1} on $\ty$ and \smash{$\tD$}.
\end{enumerate}
\end{algorithm}

If $X$ does not have full column rank (note that this is necessarily
the case when $p>n$), then a path following approach is still
possible, but is substantially more complicated---see \citet{genlasso}
for a discussion.   An easier fix (than deriving a new path algorithm)
is to simply add a term $\epsilon \|\beta\|_2^2$ to the
criterion in \eqref{eq:genlasso}, where $\epsilon$ is a small
constant.  This new criterion can be written in standard generalized
lasso form, with an augmented and full column rank predictor matrix,
and therefore we can apply Algorithm \ref{alg:dualpath2} to compute
the solution path.

\subsection{Implementation overview}
\label{sec:implementation}

We give a summary of the various implementations of Algorithm
\ref{alg:dualpath1} (and Algorithm \ref{alg:dualpath2}) presented in
this article.

\subsubsection{The signal approximator case, $X=I$}

As before, we first address the case
$X=I$. For an arbitrary penalty matrix $D$,
a somewhat naive implementation of Algorithm \ref{alg:dualpath1} would
just solve the sequence of least squares problems in Step 3(a)
independently, as the algorithm moves from one iteration to the next.
Denoting $r=m-|\cB|$, so
that $D_{-\cB}$ is an $r\times n$ matrix, each iteration here would
require $O(r^2 n)$ operations if $r\leq n$, or $O(rn^2)$
operations if $r>n$.
A smarter approach would be to compute a QR decomposition of $D^T$ or
$D$ (depending on the dimensions of $D$)
to solve the initial least squares problem in Step 1, and then
update this decomposition as $D_{-\cB}$ changes to solve the
subsequent problems in Step 3(a).  In this new strategy, each
iteration takes $O(rn)$ or $O(\max\{r^2,n^2\})$ operations (when
maintaining a QR decomposition of $D_{-\cB}^T$ or $D_{-\cB}$,
respectively), which improves upon the cost of the naive strategy by
essentially an order of magnitude.  In Section \ref{sec:generald} we
give the details of this more efficient QR-based implementation.

While the QR-based aproach is effective as a general tool, for
certain classes of problems it can be much better to take advantage of
the special structure of $D$. In Sections \ref{sec:trendfilter} and
\ref{sec:fusedlasso} we describe two such specialized
implementations, for the trend filtering and fused lasso problem
classes. Here the least squares problems in Algorithm
\ref{alg:dualpath1} reduce to solving banded linear systems (trend
filtering) or Laplacian linear systems (the fused lasso).  Since these
computations are much faster than those for generic dense linear
systems (i.e., the least squares problems given an arbitrary $D$), the
specialized implementations offer a considerable boost in efficiency.
Table \ref{tab:runtimes} provides a summary of the various
computational complexities (given per iteration).

\begin{table}[h]
\begin{center}
\begin{tabular}{|c|c|c|}
\hline
& $X=I$ & General $X$ \\
\hline
General $D$, $\rank(D)=m$ & $O(rn)$ & $O(rn)$ \\
\hline
General $D$, $\rank(D)<m$ & $O(\max\{r^2,n^2\})$ &
$O(\max\{r^2,n^2\})$ \\
\hline
Trend filtering & $O(r)$ & $O(r+nq^2)$ \\
\hline
Fused lasso* & $O(\max\{r,n\})$ & $O(\max\{r,n\}+nq^2)$ \\
\hline
\end{tabular}
\caption{\small\it Complexities of different
implementations of the dual path algorithm, designed to solve
different problems.  All complexities refer to a single
iteration of the algorithm.  Recall that we denote $r=m-|\cB|$, and
the complexity of an iteration is proportional to solving a linear
system in the $r\times r$ matrix $D_{-\cB} D_{-\cB}^T$. At the first
iteration, $\cB=\emptyset$, and so $r=m$; across iterations,
$\cB$ typically decreases in size by one (but not always---it can also
increase in size by one), until at some point $\cB=\{1,\ldots m\}$,
and so $r=0$.  For the complexities in the general $X$ case, we
write $q=\nuli(D_{-\cB})$ for the dimension of the null space of
$D_{-\cB}$, which is an unbiased estimate for the degrees of freedom
of the generalized lasso fit at the current iteration. Finally, in the
last row, the ``*'' marks the fact that the reported
complexities are based on not an empirical (rather than a formal)
understanding of the relevant linear system solver. Solutions
here are computed using a sparse Cholesky factorization of a Laplacian
matrix, whose runtime is not known to have a tight bound, but
empirically behaves linearly in the number of edges in the underlying
graph.}
\label{tab:runtimes}
\end{center}
\end{table}

\subsubsection{The general $X$ case}

Now we discuss the case of a general $X$ (having full column rank).
For a general penalty matrix $D$, the first step of
Algorithm \ref{alg:dualpath2} requires $O(np^2)$
operations to compute $X^+$, and then the QR-based implementation
outlined above can simply be applied to \smash{$\tD \in \R^{m\times
n}$}.  Note that, aside from the initial overhead of computing $X^+$,
the complexity per iteration remains the same (as in the signal
approximator case).

However, for the specialized implementations for trend filtering and
fused lasso problems, the adjustment for a general $X$ is not so
straightforward. Generally speaking, performing the transformation
\smash{$\tD=DX^+$} destroys any special structure present in the
penalty matrix, and hence the least squares problems in
Algorithm \ref{alg:dualpath1}, with \smash{$\tD$} in place of $D$, no
longer directly reduce to banded or Laplacian linear systems for trend
filtering or fused lasso problems, respectively.  Fortunately,
efficient, specialized implementations
for trend filtering and
fused lasso problems are still possible in the case of a general
$X$, as we show in Section \ref{sec:specialx}.
It is important to note that the implementations
here {\it do not need to compute an initial pseudoinverse of $X$},
and only ever require solving a full linear system in $X^T X$
at the very end of the path; this makes a big difference if early
termination of the path algorithm was of interest.  Again, see Table
\ref{tab:runtimes} for a list of per-iteration complexities of the
dual path algorithm for a general $X$, across various special cases.

% In the former case, each iteration of the
% specialized implementation requires $O(r)$ operations, where recall
% that $r=|\cB|$.  In the latter case, the specialized implementation
% does not admit a tight bound on the cost of each iteration (a very
% loose upper bound is $O(n^3)$), but tends to work
% efficiently in practice. It is important to
% remind the reader that, in the presence of a (full column rank) design
% matrix $X$ in problem \eqref{eq:genlasso}, the dual path algorithm
% operates on the penalty matrix $DX^+$ (in place of $D$).
% Generally speaking, multiplication
% by $X^+$ does not preserve structure of $D$, and therefore the implementation in
% Section \ref{sec:generald} for arbitary penalty matrices must be used. Hence, to be
% clear, trend filtering problems and fused lasso problems in which $X\not=I$
% (outside of the signal approximator case) cannot be handled by the specialized
% implementations in Sections \ref{sec:trendfilter} and
% \ref{sec:fusedlasso}, and for these problems we must use the general
% implementation in Section \ref{sec:generald}.

It is worth noting a few more high-level points about our analysis
and implementation choices before we concentrate on the details in
Sections \ref{sec:generald} through \ref{sec:specialx}.
First, in general, the total number of steps $T$ taken by the dual
path algorithm is not precisely understood.
The path algorithm tracks $m$ dual coordinates as they enter the
boundary set, but a coordinate can leave and re-enter the boundary set
multiple times, which means that the total number of steps $T$ can
greatly exceed $m$.  The main exception is the 1d fused lasso problem
in the signal approximator case, $X=I$, where it is known that a dual
coordinate will never
leave the boundary set once entered, and so the algorithm always takes
exactly $T=m$ steps \citep{genlasso}.
  Beyond this special case, a general
upper bound is $T \leq 3^m$ (as no pair of boundary
set and signs $\cB,s$ can be revisited throughout iterations of the
path algorithm), but this bound is very far what is observed in
practice. Further, solutions of interest can often be
obtained by a partial run of the path algorithm (i.e., terminating the
algorithm early) since the
algorithm starts at the fully
regularized end ($\lambda=\infty$) and produces less and less
regularized solutions as it proceeds (as $\lambda$ decreases).  For
these reasons, we choose to focus on the complexity of each iteration
of the path algorithm, and not its total complexity, in our analysis.

A second point concerns the choice of solvers for the linear systems
encountered across
steps of the path algorithm.  Broadly speaking, there are two
types of solvers for linear systems: direct and indirect
solvers.  Direct solvers (typically based on matrix factorizations)
return an exact solution of a linear system (exact up to
computer rounding errors---i.e., on a perfect computational platform,
a direct method would return an exact solution).  Indirect solvers
(usually based on iteration) produce an approximate solution
to within a user-specified tolerance level
$\epsilon$ (and their runtime depends on $\epsilon$, e.g., via a
multiplicative factor like $\log(1/\epsilon)$).
Indirect solvers will generally scale to much larger problem sizes than
direct ones, and hence they may be preferable if one can tolerate
approximate solutions.  In the context of the dual path algorithm,
however, the quality of solutions of the linear systems at each step
can strongly influence the accuracy of the algorithm in future
steps, as the boundary set $\cB$ is grown incrementally across
iterations.  In other words,
relying on approximate solutions can be risky because approximation
errors can accumulate along the path, in the sense that the algorithm
can make false additions to the boundary set $\cB$ that cannot really
be undone in future steps.  We therefore stick to direct solvers
in all proposed implementations of the dual path algorithm, across the
various special problem cases.

After describing the implementation strategies for a general penalty
matrix $D$, trend filtering problems, and fused lasso problems
in Sections \ref{sec:generald} through \ref{sec:specialx}, the rest of
this paper is dedicated to example applications the path algorithm, in
Section \ref{sec:crime}, and an empirical evaluation of the various
implementations, in Section \ref{sec:timings}.

\section{QR-based implementation for a general $D$}
\label{sec:generald}

This section considers a general penalty matrix
$D \in \R^{m\times n}$.  We assume without a loss of generality
that $X=I$; recall that a general (full column rank)
matrix $X$ contributes an additional $O(np^2)$ operations for the
computation of $X^+$, but changes nothing else---see Algorithm
\ref{alg:dualpath2}.  Hence we focus on Algorithm
\ref{alg:dualpath1}, and our
strategy is to use a QR decomposition to solve the least squares problems
at each iteration, and update it efficiently as rows are removed from or added to
$D_{-\cB}$. Appendix \ref{app:qrls} reviews the QR decomposition and how it
can be used to compute minimum $\ell_2$ norm solutions of least
squares problems.
Appendices \ref{app:upfull} and \ref{app:updef} describe techniques
for efficiently updating the QR decomposition, after rows or columns
have been added or removed.  These techniques save essentially an
order of magnitude in computational work when compared to computing
the QR decomposition anew. All of the computational complexities cited
in the following sections are verified in these appendices (a word of
warning to the reader: the roles of $m$ and $n$ are not the same in
the appendices as they are here).
%---in fact, they are exactly reversed).

We present two strategies: one that computes and maintains a QR decomposition
of $D^T$, and another that does the same for $D$. The second strategy can
handle all penalty matrices $D \in \R^{m\times n}$, regardless
of the dimensions and rank. On the other hand, the first strategy only applies to
matrices $D$ for which $m \leq n$ and $\rank(D)=m$, but is more efficient (than
the second strategy) in this case. We call the first strategy
the ``wide strategy'', and the second the ``tall strategy''.
After describing these two strategies, we make comparisons in terms of
computational order.

\subsection{The wide strategy}
\label{sec:wide}

If $m\leq n$, then we first compute the QR decomposition $D^T = QR$,
where $Q \in \R^{n\times n}$ is orthogonal and $R\in\R^{m\times n}$ is
of the form
\begin{equation*}
R = \left[\begin{array}{c} R_1 \\ 0 \end{array}\right],
\end{equation*}
where the top block $R_1 \in \R^{m\times m}$ has all zeros below its diagonal.
Computing this decomposition takes $O(m^2n)$ operations.
If one or more of the diagonal elements of $R_1$ is zero, then $\rank(D)<m$;
in this case, we skip ahead to the tall strategy (covered in the next section).
Otherwise, $R_1$ has proper upper triangular form (all nonzero diagonal
elements), which means that $\rank(D)=m$, and we
proceed with the wide strategy, outlined below.

\begin{itemize}
\item{\it Step 1.}
We first compute the minimizer $\hu$ of $\|y-D^T u\|_2^2$ (note that since
$\rank(D)=m$, this minimizer is unique). Using the QR decomposition $D^T=QR$,
this can be done in $O(mn)$ operations (Appendix \ref{app:qrfull}).

\item{\it Step 3(a).} Now we compute the minimizers
$\hat{a}$ and \smash{$\hat{b}$} of the two least squares criterions
$\|y-D_{-\cB}^T a\|_2^2$ and $\|D_\cB^T s - D_{-\cB}^T b\|_2^2$,
respectively.
The set $\cB$ has changed by one element from the previous iteration
(thinking of the boundary set as being empty in the initial least
squares problem of Step 1). By construction, we have a decomposition of
$D_{-\cB'}^T$ for the old boundary set $\cB'$ (this is initially a
decomposition of $D^T$), and as $\cB$ and $\cB'$ differ by one element,
$D_{-\cB}^T$ and $D_{-\cB'}^T$ differ by one column. Hence we can update
the QR decomposition of $D_{-\cB'}^T$ to obtain one of $D_{-\cB}^T$,
in $O(r n)$ operations, where $r=m-|\cB|$ (Appendix
\ref{app:upfullcol}), and use this to solve the two least squares
problems, in another $O(rn)$ operations.
\end{itemize}

\subsection{The tall strategy}
\label{sec:tall}

The tall strategy is used when either $m>n$, or $m \leq n$ but $D$ is
row rank deficient (which would have been detected at the beginning of
the wide strategy). We begin by computing a QR decomposition of $D$
of the special form $DPG=QR$, where
$P\in\R^{n\times n}$ is a permutation matrix,
$G\in\R^{n\times n}$ is an orthogonal matrix of Givens rotations,
$Q\in\R^{m\times m}$ is orthogonal, and $R\in\R^{m\times n}$
decomposes as
\begin{equation*}
R = \left[\begin{array}{cc} 0 & R_1 \\ 0 & 0\end{array}\right].
\end{equation*}
Here $R_1 \in \R^{k\times k}$, and $k=\rank(D)$. This special QR
decomposition, which we refer to as the {\it rotated QR decomposition},
can be computed in $O(mnk)$ operations (Appendix \ref{app:mlst}).  The
steps taken by the tall strategy are as follows.

\begin{itemize}
\item{\it Step 1.} We compute the minimum $\ell_2$ norm minimizer
$\hu$ of $\|y-D^T u\|_2^2$, exploiting the special form of rotated QR
decomposition $DPG=QR$ (more precisely, the special form of the $R$
factor). This requires $O(n\cdot\max\{m,n\})$ operations
(Appendix \ref{app:mlst}).

\item{\it Step 3(a).} Now we seek the minimum $\ell_2$ norm
minimizers $\hat{a}$ and \smash{$\hat{b}$} of
$\|y-D_{-\cB}^T a\|_2^2$, respectively
$\|D_\cB^T s - D_{-\cB}^T b\|_2^2$.
We have a rotated QR decomposition of $D_{-\cB'}$, where
$\cB'$ is the boundary set in the previous iteration (thought of
as $\cB'=\emptyset$ in Step 1, so initially this decomposition is
simply $DPG=QR$). As the current boundary set $\cB$ and the old
boundary set $\cB'$ differ by one element, $D_{-\cB}$ and $D_{-\cB'}$
differ by one row, and we can update the rotated QR decomposition of
$D_{-\cB}$ to form a rotated QR decomposition of $D_{-\cB}$, in
$O(\max\{r^2,n^2\})$ operations, for $r=m-|\cB|$ (Appendix
\ref{app:updefcol}). Computing the appropriate minimum $\ell_2$ norm
solutions then takes $O(n\cdot\max\{r,n\})$ operations.
\end{itemize}

\subsection{Computational complexity comparisons}
\label{sec:compcubed}

In the wide strategy, the initial work requires $O(m^2 n)$ operations,
and each subsequent iteration $O(rn)$ operations. Meanwhile,
for the same problems, the naive strategy (which, recall, simply
solves all least squares problems encountered in Algorithm
\ref{alg:dualpath1} separately) performs $O(r^2n)$ operations
per iteration, which is an order of magnitude larger.

The comparison for
the tall strategy is similar, but strictly speaking not quite as favorable.
The initial work for the tall strategy requires
$O(mn\cdot\min\{m,n\})$ operations, and subsequent iterations require
$O(\max\{r^2,n^2\})$ operations.
The naive strategy uses $O(rn\cdot\min\{r,n\})$ operations per
iteration, which is an order of magnitude larger if $r = \Theta(n)$, but
not if $r$ and $n$ are of drastically different sizes. E.g.,
near the end of the path (where $r=m-|\cB|$ is quite small compared to
$n$),
iterations of the tall strategy can actually be less efficient than
the naive implementation. A simple fix is to switch over to the naive
strategy when $r$ becomes small enough. In practice, the start
of the path is usually of primary interest, and the tall strategy is
much more efficient than the naive one.

In summary, if $T$ denotes the total number of iterations taken by the
algorithm, then the total complexity of the QR-based implementation
described in this section is
\begin{equation*}
\begin{array}{ll}
O(m^2 n + T mn)
&\;\;\;\text{if}\;\, m\leq n \;\,\text{and}\;\, \rank(D)=m,
\smallskip \\
O(m^2 n + T n^2)
&\;\;\;\text{if}\;\, m\leq n \;\,\text{and}\;\, \rank(D)<m,
\smallskip \\
O(m n^2 + T m^2)
&\;\;\;\text{if}\;\, m> n.
\end{array}
\end{equation*}
We remark that work of \citet{genlasso} alluded to the
implementation described in this section, but did not give any
details. This latter work also reported a computational complexity
for such an implementation, but contained a typo, in that it
essentially mixes up the complexities for the cases $m\leq n$
and $m>n$.

\section{Specialized implementation for trend filtering, $X=I$}
\label{sec:trendfilter}

We describe a specialized implementation for trend
filtering.  Recall that for such a class of problems, we have
$D=D^{(k+1)}$, the $(p-k-1) \times p$ discrete derivative operator of
order $k+1$, for some fixed integer $k \geq 0$.  These operators are
defined as
\begin{align}
\label{eq:d1}
D^{(1)} &= \left[\begin{array}{rrrrrr}
-1 & 1 & 0 & \ldots & 0 & 0 \\
0 & -1 & 1 & \ldots & 0 & 0 \\
\vdots & & & & & \\
0 & 0 & 0 & \ldots & -1 & 1
\end{array}\right], \\
\label{eq:dk+1}
D^{(k+1)} &= D^{(1)} \cdot D^{(k)}
\;\;\; \text{for} \;\, k=1,2,3,\ldots.
\end{align}
In the signal approximator case, $X=I$, trend filtering can be viewed
as a nonparametric regression estimator, producing piecewise
polynomial fits of a prespecified order $k\geq 0$, and having
favorable adaptivity properties \citep{trendfilter}.  We focus on the
$X=I$ case here, and argue that trend filtering estimates
can be computed quickly via the dual path algorithm. The case of a
general $X$ requires a more sophisticated implementation and is
handled in Section \ref{sec:specialx}.

The analysis for trend filtering is actually quite straightforward: the
key point is that discrete difference operators
as defined in \eqref{eq:d1}, \eqref{eq:dk+1} are banded matrices with
full row rank.  In particular, $D^{(k+1)}$ has bandwidth $k+2$, and
this makes $D^{(k+1)} (D^{(k+1)})^T$ an invertible $(n-k-1) \times
(n-k-1)$ banded matrix of bandwidth $2k+3$, so we can solve the
initial least squares problem in Step 1 of Algorithm
\ref{alg:dualpath1}, i.e., solve the banded linear system
\begin{equation*}
D^{(k+1)}(D^{(k+1)})^T u = D^{(k+1)}y,
\end{equation*}
in $O(nk^2)$ operations, using a banded Cholesky decomposition of
$D^{(k+1)}(D^{(k+1)})^T$ (see Section 4.3 of \citet{gvl}).
Further, for an arbitrary boundary set
$\cB \subseteq \{1,\ldots n-k-1\}$, the matrix
\smash{$D^{(k+1)}_{-\cB} (D^{(k+1)}_{-\cB})^T$} is an $r \times r$
invertible matrix with bandwidth $2k+3$, where $r=n-k-1-|\cB|$,
and hence the two least squares problems in Step 3(a) of Algorithm
\ref{alg:dualpath1}, i.e., the two linear systems
\begin{equation*}
D^{(k+1)}_{-\cB}(D^{(k+1)}_{-\cB})^T a = D^{(k+1)}_{-\cB} y
\;\;\;\text{and}\;\;\;
D^{(k+1)}_{-\cB}(D^{(k+1)}_{-\cB})^T b = D^{(k+1)}_{-\cB}
(D^{(k+1)}_\cB)^T s,
\end{equation*}
can be solved in $O(rk^2)$ operations.  Since $k$ is a constant (it
is given by the order of the desired piecewise polynomial to be fit),
we see that each iteration in this implementation of the dual path
algorithm requires $O(r)$ operations, i.e., linear time in the number
of interior (non-boundary) coordinates, as listed in Table
\ref{tab:runtimes}.

The banded Cholesky decomposition of
\smash{$D_{-\cB}^{(k+1)}(D_{-\cB}^{(k+1)})^T$} provides a
very fast way of solving the above linear systems, both in terms of
its theoretical complexity and practical performance.  Yet, we
have found that solving the linear systems (i.e., the corresponding
least squares problems) with a sparse QR
decomposition of \smash{$(D_{-\cB}^{(k+1)})^T$} is essentially just
as fast in practice, even though this approach does not yield a
competitive worst-case complexity (since \smash{$D_{-\cB}^{(k+1)}$}
itself is not necessarily banded).  Importantly, the QR approach
delivers solutions with better numerical accuracy, due to the fact
that it operates on \smash{$D_{-\cB}^{(k+1)}$} directly, rather than
\smash{$D_{-\cB}^{(k+1)} (D_{-\cB}^{(k+1)})^T$}, whose condition number
is the square of that of \smash{$D_{-\cB}^{(k+1)}$} (see Section
5.3.8 of \citet{gvl}).  For this reason, it can be preferable to use the
sparse QR decomposition in practical implementations; this is the
strategy taken by R package {\tt genlasso}, which uses a particular
sparse QR algorithm of \citet{sparseqr}.

We remark that neither of the banded Cholesky nor sparse QR
approaches proposed here utilize information between the linear
systems across iterations, i.e., we do not maintain a single matrix
decomposition and update it at every iteration.  A successful updating
scheme of this sort would only add to the efficiency of the (already
highly efficient) proposals above.  But it is important to
mention that, in general, updating a sparse matrix decomposition
demands great care; standard updating rules intended for dense matrix
decompositions (e.g., as described in
Appendix \ref{app:upfull} for the QR decomposition) do not work
well in combination with sparse matrix decompositions, since they are
typically based on operations (e.g., Givens rotations) that can create
``fill-in''---the unwanted transformation of zero elements to nonzero
elements in factors of the decomposition.  Investigating
sparsity-maintaining update schemes is a topic for future work.

\section{Specialized implementation for the fused lasso, $X=I$}
\label{sec:fusedlasso}

This section derives a specialized implementation for fused lasso
problems, where the components of $\beta \in \R^p$
correspond to nodes on some underlying graph $G$, with undirected
edge set $E \subseteq \{1,\ldots p\}^2$.  If $E$ has $m$ edges,
written as $E=\{e_1,\ldots e_m\}$, then the fused lasso penalty matrix
$D$ has dimension $m\times p$.  Specifically, if the $\ell$th edge is
$e_\ell=(i,j)$, then recall that the $\ell$th row of $D$ is given by
\begin{equation*}
D_{\ell k} = \begin{cases}
-1 & k=i \\
1 & k=j \\
0 & \text{otherwise}
\end{cases},
\;\;\; k=1,\ldots p.
\end{equation*}
(In the above, the signs are arbitrary; we could have just as well
written $D_{\ell i}=1$ and $D_{\ell j}=-1$.) In graph theory, the
matrix $D$ is known as the oriented incidence matrix of the
undirected graph $G$.  For simplification in what follows, we will
assume that $X=I$; Section \ref{sec:specialx} relaxes this assumption,
but uses a more complex implementation plan.

As we have seen, Steps 1 and 3(a) of Algorithm
\ref{alg:dualpath1} reduce to solving to linear systems of the
form $DD^T x = Dc$ and \smash{$D_{-\cB}D_{-\cB}^Tx = D_{-\cB} c$},
respectively.
With $D$ the oriented incidence matrix of a
graph, the matrices $DD^T$ and \smash{$D_{-\cB}D_{-\cB}^T$} are highly
sparse, so one might guess that it is easy to execute
such steps efficiently.  A substantial complication, however, is that
we require the minimum $\ell_2$ norm solutions of
these generically underdetermined linear systems (note, e.g., that
$DD^T$ is rank deficient when the number of edges $m$ in the
underlying graph exceeds the number of nodes $n$, and an analogous
story holds for $D_{-\cB}D_{-\cB}^T$).  For a sparse
underdetermined linear system, it is typically
possible to find an arbitrary solution---\citet{gvl} call this a basic
solution---in an efficient manner, but computing the solution with the
minimum $\ell_2$ norm is generally much more difficult.

The main insight that we contribute in this section is a
strategy for obtaining the minimum $\ell_2$ norm solution of
\begin{equation}
\label{eq:ddt}
DD^T x = Dc
\end{equation}
from a basic solution of
\begin{equation}
\label{eq:dtd}
D^T D z = d,
\end{equation}
for some $d$.
The same strategy applies to the linear problems in future iterations
with $D_{-\cB}$ taking the place of $D$.  In fact, our proposed
strategy does not place any assumptions on $D$; its only real
constraint is that right-hand side vector $d$ in
\eqref{eq:dtd} is defined by a projection onto $\nul(D)$, the
null space of $D$, so this projection operator must be readily
computable in order for the overall strategy to be effective.
Fortunately, this is the case for fused lasso problems, as the
projections onto $\nul(D)$ and $\nul(D_{-\cB})$ can be done in
closed-form, via a simple averaging calculation.

Next, we precisely describe the relationship between the minimum
$\ell_2$ norm solution of \eqref{eq:ddt} and solutions of
\eqref{eq:dtd}.  This leads to alternate expressions for the
quantities \smash{$\hu$} and \smash{$\hat{a},\hat{b}$} in Steps 1 and
3(a) of the dual path algorithm, for a general
matrix $D$.  By following such alternate representations, we then
derive a specialized implementation for fused lasso problems.

\subsection{Alternative form for Steps 1 and 3(a) in Algorithm
  \ref{alg:dualpath1}}
\label{sec:alg1alt}

We present a simple lemma, relating the solutions of
\eqref{eq:ddt} and \eqref{eq:dtd}.

\begin{lemma}
\label{lem:dtd}
For any matrix $D$, the minimum $\ell_2$ norm solution $x^*$ of the
linear system \eqref{eq:ddt} is given by $x^* = Dz$,
where $z$ is any solution of the linear system \eqref{eq:dtd}, and
$d=P_{\row(D)} c = (I-P_{\nul(D)})c$.
\end{lemma}

\begin{proof}
We can express the minimum $\ell_2$ norm solution
of \eqref{eq:ddt} as
\begin{equation*}
x^* = (D^T)^+c = (D^+)^T c = D (D^T D)^+ c,
\end{equation*}
using the fact that pseudoinverse and transpose operations
commute.  Now $z^* = (D^T D)^+c$ is the minimum $\ell_2$
norm solution of the linear system \eqref{eq:dtd}, provided that
$d=P_{\row(D)} c$.   Hence $x^* = Dz^*$.  But any
solution $z$ of \eqref{eq:dtd} has the form $z=z^* + \eta$, where
$\eta \in \nul(D)$, and therefore also $Dz = Dz^* + D\eta = x^*$.
\end{proof}

As a result, we can now reexpress the computation of $\hu$ and
\smash{$\hat{a},\hat{b}$} in Steps 1 and 3(a), respectively, of
Algorithm \ref{alg:dualpath1} as follows.
\begin{itemize}
\item {\it Step 1.}  Compute $v=(I-P_{\nul(D)})y$, solve the linear
  system $D^T D z = v$, and set $\hu = Dz$.
\item {\it Step 3(a).}  Compute \smash{$v=(I-P_{\nul(D_{-\cB})}) y$}
and \smash{$w = (I-P_{\nul(D_{-\cB})}) D_\cB^T s$}, solve the linear
systems \smash{$D_{-\cB}^T D_{-\cB} z = v$} and
\smash{$D_{-\cB}^T D_{-\cB} x = w$}, and
then set \smash{$\hat{a} = D_{-\cB} z$} and
\smash{$\hat{b} = D_{-\cB} x$}.
\end{itemize}
For an arbitrary $D$, using these alternate forms of the steps does
not necessarily provide a computational advantage over our existing
approach in Section \ref{sec:tall}.  For one, at each step we must
compute a projection onto $\nul(D)$ or $\nul(D_{-\cB})$, which is
generically just as difficult as maintaining a (rotated) QR
decomposition to compute the minimum $\ell_2$ norm solution of a
linear system in $DD^T$ or \smash{$D_{-\cB}D_{-\cB}^T$} (as covered
in Section \ref{sec:tall}).  A second point is that $D$ must be sparse
in order for there to be a genuine difference between computing basic
solutions and minimum $\ell_2$ norm solutions of linear systems
involving $D$.  However, in special
cases, e.g., the fused lasso case, working from the alternate forms of
Steps 1 and 3(a) given above can make a big difference in terms of
efficiency.

\subsection{Laplacian-based implementation for fused lasso problems}
\label{sec:lapfl}

The alternate forms of Steps 1 and 3(a) given in the previous
section have particularly nice translations for fused
lasso problems, with $D \subseteq \R^{m\times n}$ being the oriented
incidence matrix of a graph $G$.  In this case,
projections onto $\nul(D)$ and $\nul(D_{-\cB})$, as well as basic
solutions of linear systems in $D^T D$ and \smash{$D_{-\cB}^T
D_{-\cB}$}, can both be computed efficiently.

\subsubsection{Null space of the oriented incidence matrix}

We address the null space computations first.  It is not hard to
see that here the null space of $D$ is spanned by the indicators of
connected components $C_1,\ldots C_r$ of the graph $G$, i.e.,
\begin{equation*}
\nul(D) = \mathrm{span}\{1_{C_1},\ldots 1_{C_r}\},
\end{equation*}
where each $1_{C_j} \in \R^n$, and has components
\begin{equation*}
(1_{C_j})_i = \begin{cases}
1 & i \in C_j \\
0 & \text{otherwise}
\end{cases},
\;\;\; i=1,\ldots n.
\end{equation*}
Hence, projection onto $\nul(D)$ is simple and efficient, and
is given by componentwise averaging,
\begin{equation*}
(P_{\nul(D)} x)_i = \frac{1}{|C_j|} \sum_{\ell \in C_j} x_\ell
\;\;\;\text{where}\;\, C_j \ni i, \;\;\; \text{for each}\;\,
i=1,\ldots n.
\end{equation*}
For an arbitrary subset $\cB$ of $\{1,\ldots m\}$, note that
$D_{-\cB}$ is the oriented incidence matrix of the graph $G_{-\cB}$,
which denotes the graph $G$ after we delete the edges corresponding
to $\cB$ (in other words, $G_{-\cB}$ is the graph with nodes
$\{1,\ldots n\}$ and edges $\{e_\ell, \, \ell \notin \cB\}$).  Therefore
the same logic as above applies to projection onto $\nul(D_{-\cB})$:
it is given by componentwise averaging within the connected
components of $G_{-\cB}$,
\begin{equation*}
(P_{\nul(D_{-\cB})} x)_i = \frac{1}{|C_j|} \sum_{\ell \in C_j} x_\ell
\;\;\;\text{where}\;\, C_j \ni i, \;\;\; \text{for each}\;\,
i=1,\ldots n,
\end{equation*}
and $C_j$ now denotes the $j$th connected component of $G_{-\cB}$.

\subsubsection{Solving Laplacian linear systems}
\label{sec:laplin}

Now we discuss computing basic solutions of linear systems in $D^T D$
or \smash{$D_{-\cB}^T D_{-\cB}$}.  As $D$ is the oriented incidence
matrix of $G$, this makes $D^T D$ the Laplacian matrix of $G$;
similarly \smash{$D_{-\cB}^T D_{-\cB}$} is the Laplacian matrix of the
graph $G_{-\cB}$.   The Laplacian linear system is a well-studied
topic in computer science; see, e.g.,
\citet{lapsol} for a nice review paper.  In principle, any fast solver
can be used for the Laplacian linear systems in Steps 1 and 3(a) of
the path algorithm, as presented in Section \ref{sec:alg1alt}.
However, in practice, using indirect or iterative solvers (which
return approximate solutions, according to a user-specified tolerance
level for approximation) for the linear systems at each step can
cause practical issues with the path algorithm, as explained in the
introduction.  For the current setting, this precludes the use of the
extremely fast indirect algorithms for Laplacian linear systems that
have been recently developed by the theoretical computer science
community (again see \citet{lapsol}, and references therein).  We
focus instead on a simple direct solver.

Let $L$ denote the Laplacian matrix of an arbitrary graph.  If the
graph has $r$ connected components, then (modulo a reordering of its
rows and columns) $L$ can be expressed as
\begin{equation}
\label{eq:ldecomp}
L = \left[\begin{array}{cccc}
L_1 & 0 & \ldots & 0 \\
0 & L_2 & \ldots & 0 \\
\vdots & & & \\
0 & 0 & \ldots & L_r
\end{array}\right],
\end{equation}
i.e., a block diagonal matrix with $r$ blocks.  Therefore, the
Laplacian linear system $Lx = b$ reduces to solving $r$ separate
systems $L_j x_j = b_j$ (here we have decomposed $b=(b_1,\ldots
b_r)$ according to the same block structure), and then concatenating
$x=(x_1,\ldots x_r)$ to recover the original solution.

Note that each matrix $L_j$, $j=1,\ldots r$ is the Laplacian matrix
of a fully connected subgraph; this means that the null space of
$L_j$ is exactly 1-dimensional (it is spanned by the vector of all
1s), and that the linear system $L_jx_j = b_j$ is underdetermined.
The following lemma provides a remedy.

\begin{lemma}
\label{lem:laplin}
Let $L$ be the Laplacian matrix of a connected graph with $n$ nodes.
Write $L$ as
\begin{equation*}
L = \left[\begin{array}{cc}
A & c \\
c^T & d
\end{array}\right],
\end{equation*}
where $A \in \R^{(n-1)\times (n-1)}$, $c\in\R^{n-1}$, and $d\in\R$.
Then for any $b\in\col(L)\subseteq\R^n$, the Laplacian
linear equation
\begin{equation}
\label{eq:llin}
Lx = b
\end{equation}
is solved by $x=(x_{-n},0)$, where $x_{-n} \in \R^{n-1}$ is the
unique solution of
\begin{equation}
\label{eq:alin}
Ax_{-n} = b_{-n},
\end{equation}
with $b_{-n}\in\R^{n-1}$ containing the first $n-1$ components of $b$,
or in other words, $x_{-n}=A^{-1}b_{-n}$.
\end{lemma}

\noindent
{\it Remark.} The ordering of the $n$ nodes in the graph does not
matter. Therefore, to solve $Lx=b$, we can actually consider the
submatrix $A\in\R^{(n-1)\times(n-1)}$ formed by excluding the $i$th
row and column from $L$, and solve the subsystem $Ax_{-i}=b_{-i}$,
for any $i \in \{1,\ldots n\}$.

\begin{proof}
Since the graph is connected, the null space of $L$
is 1-dimensional, and spanned by $(1,\ldots 1)\in\R^n$.
Hence $\rank(L)=n-1$, and the rank of the first $n-1$ columns of $L$
is at most $n-1$,
\begin{equation*}
\rank\left(\left[\begin{array}{c} A \\ c^T \end{array}\right]\right)
\leq n-1.
\end{equation*}
Assume that
\begin{equation}
\label{eq:rank}
\rank\left(\left[\begin{array}{c} A \\ c^T \end{array}\right]\right)
= n-1.
\end{equation}
Then $L$ and its first $n-1$ columns have the same image, so given any
$b$ in this image, there must exist a solution $x_{-n} \in \R^{n-1}$ in
\begin{equation}
\label{eq:im}
\left[\begin{array}{c} A \\ c^T \end{array}\right] x_{-n} = b,
\end{equation}
which yields a solution of $L
x=b$ with $x=(x_{-n},0)$. Moreover,
the solution $x_{-n}$ of \eqref{eq:im} is unique (by \eqref{eq:rank}), and
to find it we can restrict our attention to the first $n-1$ equalities,
$Ax_{-n}=b_{-n}$.

Therefore it suffices to prove the rank assumption \eqref{eq:rank}.
For this, we can equivalently prove that the first $n-1$ columns
of the oriented incidence matrix $D$ of the graph have rank $n-1$. Let
$D'$ denote these first $n-1$ columns, and let $E$ denote the edge set
of the graph. Note that, for each $(i,n) \in E$, there is a
corresponding row of $D'$ with a single $1$ or $-1$ in the $i$th
component.  Suppose that $D'v=0$; then immediately we have $v_i=0$
for any $i$ such that $(i,n) \in E$. But this implies that $v_j=0$ for all
$j$ such that $(j,i)\in E$ and $(i,n)\in E$, and repeating this argument,
we eventually conclude that
$v_i=0$ for all $i=1,\ldots n-1$, because the graph is connected.
We have shown that $\nul(D')=\{0\}$, and so $\rank(D')=n-1$, as
desired.
\end{proof}

The message of Lemma \ref{lem:laplin} is that, for a fully connected
graph and the Laplacian linear system \eqref{eq:llin}, we can solve
this system by instead solving a smaller system \eqref{eq:alin},
formed by removing (say) the last row and column of the Laplacian
matrix.  The latter system \eqref{eq:alin} can be solved efficiently
because it is sparse and nonsingular (e.g., using a sparse Cholesky
decomposition).  Of course, for the
linear system $Lx=b$ with $L$ a generic graph Laplacian, we apply
Lemma \ref{lem:laplin} to each subsystem $L_j x_j = b_j$, $j=1,\ldots
r$, after decomposing $L$ according to its $r$ connected components,
as in \eqref{eq:ldecomp}.

\subsubsection{Tracking graph connectivity across iterations}

We finish describing the specialized implementation for fused lasso
problems.  As explained earlier, the dual path algorithm repeatedly
computes projections onto $\nul(D)$ or \smash{$\nul(D_{-\cB})$}, and
solves linear systems in the Laplacian
$L=D^T D$ or \smash{$L=D_{-\cB}^T D_{-\cB}$}, across the Steps 1 and
3(a) described in Section \ref{sec:alg1alt}.  To utilize the approaches
outlined above, each step requires finding the
connected components of the graph $G$ or $G_{-\cB}$.  Across
successive iterations, these graphs are highly related---from one
iteration to the next, $G_{-\cB}$ only changes by the addition or
deletion of one edge (since $D_{-\cB}$
only changes by the addition or deletion of one row).  Therefore we
can easily check whether adding or deleting such an edge $e$
has changed the connectivity of the graph, by running a
breadth-first search (or depth-first search) from one of the nodes
incident to $e$.  Incorporating this idea into the path following
strategy finalizes our specialized implementation for the fused lasso,
which we summarize below.

\begin{itemize}
\item {\it Step 1.}  Compute the connected components of the
  graph $G$ (corresponding to the oriented incidence matrix
  $D$).  Compute $v=(I-P_{\nul(D)})y$ by centering $y$
  over each connected component.  Solve the Laplacian linear system
  $D^T D z = v$ by decomposing into linear subsystems over each
  connected component, and applying Lemma \ref{lem:laplin} to each
  subsystem. Set $\hu=Dz$.

\item {\it Step 3(a).}  Find the connected components of $G_{-\cB}$ by
  running breadth-first (or depth-first) search, starting at a node
  that is incident to the edge added or deleted at the last iteration.
  Compute the projections \smash{$v=(I-P_{\nul(D_{-\cB})}) y$}
  and \smash{$w = (I-P_{\nul(D_{-\cB})}) D_\cB^T s$} by centering $y$
  and \smash{$D_\cB^T s$} over each connected component.  Solve the
  Laplacian linear systems \smash{$D_{-\cB}^T D_{-\cB} z = v$} and
  \smash{$D_{-\cB}^T D_{-\cB} x = w$} by decomposing into smaller
  subsystems over each connected component, and then applying Lemma
  \ref{lem:laplin}.  Set \smash{$\hat{a} = D_{-\cB} z$} and
  \smash{$\hat{b} = D_{-\cB} x$}.
\end{itemize}

For each Laplacian linear subsystem encountered (given by
decomposing the Laplacian linear systems at each step across
connected components), the {\tt genlasso} R package uses a
sparse Cholesky decomposition on the reduced system \eqref{eq:alin},
as prescribed by Lemma \ref{lem:laplin}.  In particular, it employs a
sparse Cholesky algorithm of \citet{sparsechol} (see also the
references therein).  Unfortunately, this sparse Cholesky approach
does admit a tight bound on the compexity of solving
\eqref{eq:alin}, but empirically it is quite efficient, and the number
of operations scales linearly in the number of edges in the
subgraph (provided that this exceeds the number of nodes).  This means
that the complexity of solving a full Laplacian linear system is
approximately linear in the number of edges in the graph, and so, each
iteration of the dual path algorithm requires approximately
$O(\max\{r,n\})$ operations, where $r=m-|\cB|$ is the number of edges
in $G_{-\cB}$, and $n$ is the number of nodes.

\subsection{Extension to sparse fused lasso problems}
\label{sec:lapsfl}

The specialized fused lasso implementation of the last section can be
extended to cover the sparse fused lasso problem, where the
penalty matrix $D$ is now the oriented incidence matrix of a graph
$D^{(G)}$ with a constant multiple of the identity appended to its
rows, i.e.,
\begin{equation*}
D = \left[\begin{array}{c} D^{(G)} \\ \alpha I \end{array}\right],
\end{equation*}
so that
\begin{equation*}
\|D\beta\|_1 = \sum_{(i,j) \in E} |\beta_i-\beta_j| + \alpha
\|\beta\|_1,
\end{equation*}
for some edge set $E$ and fixed constant $\alpha > 0$.  For
brevity, we state without proof here results on the appropriate null
space projections and linear systems.
First, projecting onto $\nul(D)$ is trivial, because $\nul(D)=\{0\}$
(due to the fact that $\alpha>0$).  Consider projection onto
$\nul(D_{-\cB})$.  If there are $m$ edges in the underlying graph $G$,
then $D$ is $(m+n) \times n$, with its first $m$ rows corresponding to
the edges, and its last $n$ rows corresponding to the nodes.
As in \citet{genlasso}, we can partition the
boundary set $\cB$ accordingly, writing $\cB=\cB_1 \cup (m+\cB_2)$,
where $\cB_1 \subseteq \{1,\ldots m\}$ and $\cB_2 \subseteq
\{1,\ldots n\}$.   Furthermore, we can think of $D_{-\cB}$ as
corresponding to a subgraph $G_{-\cB}$ of $G$, defined by restricting
both of its edge and node sets, as follows:
\begin{itemize}
\item we first delete all edges of $G$ that correspond to $\cB_1$,
  yielding $G_{-\cB_1}$;
\item we then delete all nodes of $G_{-\cB_1}$ that are in $\{1,\ldots
  n\} \setminus \cB_2$, and all of their connected nodes, yielding
  $G_{-\cB}$.
\end{itemize}
The projection operator onto $\nul(D_{-\cB})$
assigns a zero to each coordinate that does not correspond to
a node in $G_{-\cB}$, and otherwise performs averaging within each of
the connected components. More formally, $(P_{\nul(D_{-\cB})} x)_i =
0$ if $i$ is not a node of $G_{-\cB}$, and otherwise
\begin{equation*}
(P_{\nul(D)} x)_i = \frac{1}{|C_j|} \sum_{\ell \in C_j} x_\ell
\;\;\;\text{where} \;\, C_j \ni  i,
\end{equation*}
and $C_j$ is the $j$th connected component of $G_{-\cB}$.

As for solving linear systems in $D^T D$ or \smash{$D_{-\cB}^T
  D_{-\cB}$}, we note that
\begin{equation*}
D^T D = (D^{(G)})^T D^{(G)} + \alpha^2 I \;\;\;\text{and}\;\;\;
D_{-\cB}^T D_{-\cB} = (D^{(G)}_{-\cB_1})^T  D^{(G)}_{-\cB_1} +
\alpha^2 I_{-\cB_2}^T I_{-\cB_2}.
\end{equation*}
In either
case, the first term is a graph Laplacian, and the second term is
a multiple of the identity matrix with some of its diagonal entries
set to zero.  This means that $D^T D$ and \smash{$D_{-\cB}^T
  D_{-\cB}$} still decompose, as before, into sub-blocks over the
connected components of $G$ and \smash{$G_{-\cB_1}$}, respectively;
i.e., we can decompose both $D^T D$ and \smash{$D_{-\cB}^T  D_{-\cB}$}
as
\begin{equation*}
\left[\begin{array}{cccc}
L_1+I_1 & 0 & \ldots & 0 \\
0 & L_2+I_2 & \ldots & 0 \\
\vdots & & & \\
0 & 0 & \ldots & L_r+I_r
\end{array}\right],
\end{equation*}
where $L_1,\ldots L_r$ are Laplacian matrices corresponding to the
subgraphs of connected components, and $I_1,\ldots I_r$ are identity
matrices with some (possibly none, or all) diagonal elements set to
zero.  Hence, linear systems in $D^T D$ or \smash{$D_{-\cB}^T
  D_{-\cB}$} can be reduced to separate linear systems in $L_j+I_j$,
for $j=1,\ldots r$; for the $j$th system, if all diagonal elements of
$I_j$ are zero, then we use the strategy discussed in Section
\ref{sec:laplin} to solve the linear system in $L_j$, otherwise
$L_j+I_j$ is nonsingular and can be factored directly (using, e.g.,
again a sparse Cholesky decomposition).

For the sake of completeness, we recall a
result from \citet{pco}, which says that the sparse fused lasso
solution at any value of $\lambda$ can be computed by
solving the corresponding fused lasso problem (i.e., corresponding to
$\alpha=0$), and then soft-thresholding the output by the amount
$\alpha\lambda$.  That is, the solution path (over $\lambda$, for
fixed $\alpha$) of the sparse fused lasso problem is obtained
by just soft-thresholding the corresponding fused lasso solution
path. Given this fact, there may seem to be no reason to extend the
implementation of Section \ref{sec:lapfl} to the sparse fused lasso
setting, as we did above.  However, for a general $X$ matrix,
the simple soft-thresholding fix is no longer
applicable, and the above perspective will prove quite useful, as we
will see shortly.

%% \section{Other special cases and extensions}
%% \label{sec:other}

%% {\it Maybe mention uneven grid for, edge weights for graph, all follows
%% from the same principles? Observations weights are fine too?}

\section{Specialized implementations with a general $X$}
\label{sec:specialx}

Recall that in the presence of a (full column rank) predictor matrix
$X$, we can view the dual of the generalized lasso problem as having
the same canonical form as the dual in the signal approximator
case, but with $\ty=XX^+y$ and \smash{$\tD=DX^+$} in place of
$y$ and $D$.  The usual dual path algorithm can then be
simply run on \smash{$\ty,\tD$}, as in Algorithm \ref{alg:dualpath2}.
This is a fine strategy when $D$ is a generic penalty
matrix (Section \ref{sec:generald}).  But when $D$ is structured,
and leads to fast solutions of the appropriate linear system over
iterations of the dual path
algorithm with $X=I$ (Sections \ref{sec:trendfilter} and
\ref{sec:fusedlasso}), this structure is not in general
retained by \smash{$\tD=DX^+$}, and so
``blindly'' applying the usual path algorithm to
\smash{$\tD$} can result in a large drop in relative efficiency.

In this section, we present an
approach for carefully constructing solutions to the relevant least
squares problems, when running Algorithm \ref{alg:dualpath1} on
\smash{$\ty,\tD$} in place of $y,D$.  At a high level, our approach
solves a linear system in \smash{$\tD \in \R^{m\times n}$} using three
steps:
\begin{enumerate}
\item compute $H \in \R^{p\times q}$, whose columns are a basis for
$\nul(D)$;
\item solve a linear system in $XH \in \R^{n\times q}$;
\item solve a linear system in $D \in \R^{m\times p}$.
\end{enumerate}
In future iterations, the same strategy applies to solving linear
systems in \smash{$\tD_{-\cB} \in \R^{r \times n}$}: we repeat
the above three steps, but with $D_{-\cB}$ playing the role of $D$.
An important feature of our approach is that the matrix $XH$ in the
second step is $n \times q$, where $q$ is typically small at points of
interest along the path---we will give a more detailed explanation
shortly, but the main idea is that, at such points, linear systems
in $XH$ can be solved much more efficiently than a full linear system
in $X$ (the computational equivalent of calculating $X^+$).
Altogether, if a basis for $\nul(D)$ or
$\nul(D_{-\cB})$ is known explicitly (or can be computed easily), and
linear systems in $D$ or $D_{-\cB}$ can be solved quickly, then the
procedure outlined above can be considerably
more efficient than solving abitrary, dense linear systems in
\smash{$\tD$} or \smash{$\tD_{-\cB}$} directly.  This is the case for
both trend filtering and fused lasso problems.

Below we describe the three step procedure in detail, proving its
correctness in the context of a general matrix $D$ (and general
$X$).  After this, we discuss implementation specifics for trend
filtering and the fused lasso.

\subsection{Alternate form of computations in Algorithm
  \ref{alg:dualpath2}}

For arbitrary matrices $D,X$ (with $X$ having full column rank),
consider the problem of computing the minimum $\ell_2$ norm solution
$x^*$ of the linear system
\begin{equation}
\label{eq:dx+}
(DX^+)(DX^+)^T x = (DX^+)^T c.
\end{equation}
Our next lemma says that $x^*$ can also be characterized
as the minimum $\ell_2$ norm solution of
\begin{equation}
\label{eq:dxt}
DD^T x = D X^T d,
\end{equation}
for a suitably chosen vector $d$.

\begin{lemma}
\label{lem:dxt}
For any matrices $D,X$ (with the same number of columns) such that $X$
has full column rank, the minimum $\ell_2$ norm solution $x^*$ of
\eqref{eq:dx+} is given by the minimum $\ell_2$ norm solution of
\eqref{eq:dxt}, where $d = (I- P_{X\nul(D)}) c$.
\end{lemma}

\begin{proof}
Note that $x^* = ((DX^+)^T)^+ c$.  In general, the point $x^* = A^+ c$
can be characterized as the unique solution of the linear system $Ax =
P_{\col(X)} b$ such that $x \in \row(A)$.  (Taking $P_{\col(X)} b$,
instead of simply $b$, as the right-hand side in the linear system
here is important---the system will not be solvable if $b \notin
\col(A)$.)  Applying this logic to $A = (DX^+)^T$, we see that $x^*$
is the unique solution of
\begin{equation*}
(X^+)^T D^T x = P_{\col((DX^+)^T)} c \;\;\;\text{subject to}\;\,
x \in \row((DX^+)^T).
\end{equation*}
We have $\row((DX^+)^T) =
\col(DX^+) = \col(D)$, the last equality following since $X$ has full
column rank.  Letting {$c' = P_{\col((DX^+)^T)} c$}, the above can be
rewritten as
\begin{equation*}
(X^+)^T D^T x = c' \;\;\;\text{subject to}\;\,
x \in \col(D),
\end{equation*}
i.e., multiplying both sides by $X^T$,
\begin{equation*}
D^T x = X^T c' \;\;\;\text{subject to}\;\,
x \in \col(D),
\end{equation*}
where we again used the fact that $X$ has full column rank.  The
solution of the constrained linear system above is $x^\star = (D^T)^+
X^T c'$; in other words, we see that $x^\star$ is the minimum
$\ell_2$ norm solution of
\begin{equation*}
DD^T x = D X^T c'.
\end{equation*}
Finally, we examine {$X^T c' = X^T P_{\col((DX^+)^T)} c =
X^T P_{\row(DX^+)} c = X^T (I-P_{\nul(DX+)}) c$}.  The null space of
$DX^+$ decomposes as
\begin{align*}
\nul(DX^+)
&= \nul(X^T) + \{ z \in \col(X) : DX^+ z = 0 \} \\
&= \nul(X^T) + X\nul(D).
\end{align*}
Furthermore, the two subspaces in this decomposition are
orthogonal, so $P_{\nul(DX^+)} = P_{\nul(X^T)} + P_{X\nul(D)}$, and in
particular, {$X^T (I-P_{\nul(DX^+)}) c = X^T (I-P_{X\nul(D)}) c$},
completing the proof.
\end{proof}

If $H$ is a matrix whose columns span $\nul(D)$, then note that
projection onto $X\nul(D)$ is given by solving a least squares problem
in $XH$, namely,
\smash{$P_{X\nul(D)} c =  XH (H^T X^T X H)^{-1} H^T X^T c$.} Now,
using Lemma \ref{lem:dxt}, we can rewrite the least squares
computations in Steps 1 and 3(a) of Algorithm \ref{alg:dualpath1}
applied to \smash{$\ty,\tD$} (i.e., as would be done through Algorithm
\ref{alg:dualpath2}).
\begin{itemize}
\item {\it Step 1.}  Compute a basis $H$ for $\nul(D)$, and compute
\begin{equation}
\label{eq:xh1}
v = X^T (I-P_{X\nul(D)}) \ty = X^T y -
X^T XH (H^T X^T X H)^{-1} H^T X^T y.
\end{equation}
(Here we used the simplification $X^T \ty = X^T y$.)  Then
compute $\hu$ by solving for the minimum $\ell_2$ norm solution of the
linear system
\begin{equation}
\label{eq:easy1}
DD^T u = Dv.
\end{equation}

\item {\it Step 3(a).}  Compute a basis $H$ for
$\nul(D_{-\cB})$, and compute
\begin{align}
\label{eq:xh3a1}
v &= X^T (I-P_{X\nul(D_{-\cB})}) \ty = X^T y -
X^T XH (H^T X^T X H)^{-1} H^T X^T y, \\
\label{eq:xh3a2}
w &= X^T (I-P_{X\nul(D_{-\cB})}) \tD_\cB^T s = D_\cB^T s -
X^T XH (H^T X^T X H)^{-1} H^T D_\cB^T s.
\end{align}
(Again we used that $X^T \ty = X^T y$, and also
\smash{$X^T \tD_\cB^T s = D_\cB^T s$}.) Then compute $\hat{a}$ and
\smash{$\hat{b}$} by solving for the minimum $\ell_2$ norm solutions
of the systems
\begin{equation}
\label{eq:easy3a}
D_{-\cB}D_{-\cB}^T a = D_{-\cB} v
\;\;\;\text{and}\;\;\;
D_{-\cB}D_{-\cB}^T b = D_{-\cB} w,
\end{equation}
respectively.
\end{itemize}
For a general penalty matrix $D$, the above formulation does not
offer any advantage over applying Algorithm \ref{alg:dualpath1} to
\smash{$\ty,\tD$} directly.  But it does offer significant
advantages if the matrix $D$ is such that a basis for
$\nul(D)$ and $\nul(D_{-\cB})$ can be computed quickly, and also,
minimum $\ell_2$ norm solutions of  linear systems in $DD^T$ and
$D_{-\cB} D_{-\cB}^T$ can be computed efficiently.  In this case, we
have reduced the (generically) hard linear systems in \smash{$\tD
  \tD^T$} and \smash{$\tD_{-\cB} \tD_{-\cB}^T$} to easier ones in
$DD^T$ and $D_{-\cB}D_{-\cB}^T$, as in \eqref{eq:easy1} and
\eqref{eq:easy3a}.  Additionally, as we remarked
previously, the above steps do not require explicit computations
involving $X^+$.  Instead, the null
projections in each step require solving linear systems in
$(XH)^T XH$, as in \eqref{eq:xh1} and \eqref{eq:xh3a1},
\eqref{eq:xh3a2}.  The matrix $H$ has columns that span
$\nul(D)$ in the first iteration and span $\nul(D_{-\cB})$ in future
iterations, so $(XH)^T XH$ is $q\times q$, where $q=\nuli(D)$ or
$q=\nuli(D_{-\cB})$.  This means that $q \ll p$ at the beginning of the
path, with $q$ either increasing by one or decreasing by one at each
iteration, and only ever reaching $q=p$ when $\cB=\emptyset$ at the
end of the path. In fact, such a quantity $q$ serves as an
unbiased estimate of the degrees of freedom of the generalized lasso
estimate along the path \citep{genlasso}.  Therefore, when regularized
estimates are of interest, our focus is on the early stages of path with
$q \ll p$, in which case solving a linear system in the $q\times q$
matrix $(XH)^T XH$ is far more efficient than solving a linear system
in the $p\times p$ matrix $X^T X$ (which is what is needed in order to
apply $X^+$).

% Lemma \ref{lem:dxt} applies to offers an alternate way of solving
% \eqref{eq:dx+} that can have two significant advantages (depending
% on $D$).  First, the transformed problem \eqref{eq:dx+}  is a linear
% system in $DD^T$, so when $D$ is sparse or structured, we can
% compute the (minimum $\ell_2$ norm) solution to this system quickly,
% as argued in Sections \ref{sec:trendfilter} and \ref{sec:fusedlasso}
% for the trend filtering and fused lasso problem classes.

\subsection{Trend filtering, general $X$}

Following the alternate form of Steps 1 and 3(a) in the
last section, note that we really only need to describe the
construction of the basis matrix $H$ used in \eqref{eq:xh1} and
\eqref{eq:xh3a1}, \eqref{eq:xh3a2}, as the linear systems in
\eqref{eq:easy1} and \eqref{eq:easy3a} can then be solved by using the
sparse QR strategy outlined in Section \ref{sec:trendfilter}, for
trend filtering in the case $X=I$.

Let $D=D^{(k+1)} \in \R^{(p-k-1) \times p}$, the $(k+1)$st order
discrete difference operator
defined in \eqref{eq:d1}, \eqref{eq:dk+1}.  First we describe
$\nul(D)$.  Define $v_0 = (1,\ldots 1) \in \R^p$, and define
$v_j \in \R^p$, $j=1,2,3,\ldots$ by taking repeated cumulative
sums, as in
\begin{equation*}
(v_j)_i = \sum_{\ell=1}^{i} (v_{j-1})_\ell, \;\;\; i=1,\ldots p.
\end{equation*}
Therefore $v_1=(1,2,3,\ldots)$, $v_2=(1,3,6,\ldots)$, etc.  From
its recursive representation in \eqref{eq:d1}, \eqref{eq:dk+1}, it is
not hard to see that $\nul(D)$ is $k+1$ dimensional and spanned by
$v_0,\ldots v_k$, i.e., we can take the basis matrix $H$ to have
columns $v_0,\ldots v_k$.

Now consider $\nul(D_{-\cB})$, for an arbitrary subset
$\cB = \{i_1,\ldots i_d\} \subseteq \{1.\ldots p-k-1\}$.  One can
check that the $\nul(D_{-\cB})$ is $k+1+d$ dimensional and
spanned by $v_0,\ldots v_{k+d}$, where
$v_0,\ldots v_k$ are defined as above, and we additionally define for
$j=1,\ldots d$,
\begin{equation*}
(v_{k+j})_i = \begin{cases}
0 & \text{if}\;\, i < i_j + k + 1 \\
(v_k)_{i-i_j-k} & \text{if}\;\, i \geq i_j+k+1
\end{cases},
\;\;\; i=1,\ldots p.
\end{equation*}
Hence we take the basis matrix $H$ to have columns $v_0,\ldots
v_{k+d}$.

\subsection{Fused lasso and sparse fused lasso, general $X$}

For the fused lasso and sparse fused lasso setups, we have already
described projection onto $\nul(D)$ and $\nul(D_{-\cB})$ in Sections
\ref{sec:lapfl} and \ref{sec:lapsfl}, respectively, from which we can
readily construct a basis and populate the columns of $H$, and then
use the Laplacian-based solvers for least squares problems in
\eqref{eq:easy1} and \eqref{eq:easy3a}, as described again in Sections
\ref{sec:lapfl} and \ref{sec:lapsfl}.

To reiterate: when $D = D^{(G)} \in \R^{m\times p}$, the oriented
incidence matrix of some graph $G$, the null space of $D_{-\cB}$ for
any set $\cB \subseteq \{1,\ldots m\}$ is spanned by $1_{C_1},\ldots
1_{C_r} \in \R^p$, the indicators of connected components
$C_1,\ldots C_r$ of the graph $G_{-\cB}$. Note $G_{-\cB}$ is the
subgraph formed by removing edges of $G$ that correspond $\cB$ (i.e.,
the graph with oriented incidence matrix $D_{-\cB}$).  Therefore, in
this case, the vectors $1_{C_1},\ldots 1_{C_r}$ give the columns of
$H$. Instead, suppose that
\begin{equation*}
D = \left[\begin{array}{c} D^{(G)} \\ \alpha I \end{array}\right],
\end{equation*}
where $\alpha > 0$ and $I$ is the $p\times p$ identity matrix.  Given
any set $\cB \subseteq \{1,\ldots m+p\}$, we can partition this as
$\cB=\cB_1 \cup (m+\cB_2)$, where $\cB_1$ contains elements in
$\{1,\ldots m\}$ and $\cB_2$ elements in $\{1,\ldots p\}$.  We can
also form a subgraph $G_{-\cB}$ by removing both some of the
edges and nodes of $\cG$: we first remove the edges in $\cB_1$,
and then in what remains, we keep only the nodes that are not
connected to a node in $\cB_2$.  Writing ${C_1},\ldots {C_r}$ for the
connected components of $G_{-\cB}$, a basis for $\nul(D_{-\cB})$ can
then be obtained by appropriately padding the indicators
$1_{C_1},\ldots 1_{C_r}$ (of nodes on $G_{-\cB}$) with zeros.  That
is, for each $j$, define
\begin{equation*}
(v_j)_i =
\begin{cases}
0 & \text{if $i$ is not a node of $G_{-\cB}$} \\
(1_{C_j})_i & \text{otherwise}
\end{cases}, \;\;\; i=1,\ldots p,
\end{equation*}
and accordingly $v_1,\ldots v_r$ give a basis for $\nul(D_{-\cB})$,
and hence the columns of $H$.

%\subsection{Column rank deficiency in $X$}

% \section{Examples}
% \label{sec:examples}

% In this section we show three examples using the generalized lasso
% dual path algorithm.  The first considers a fused lasso problem over a
% custom-defined underlying graph with $X=I$; the second considers a
% fused lasso problem over another custom-defined graph with a
% non-identity predictor matrix $X$; the last considers a trend
% filtering problem with $X=I$.

\section{Chicago crime data example}
\label{sec:crime}

In this section, we give an example of the 
dual path algorithm run on a fused lasso problem with a reasonably
large, geographically-defined underlying graph.  The data comes from  
police reports made publically available by the city of Chicago, from
2001 until the present \citep{crimedata}.  These reports contain the  
date, time, type, and reported latitude and longitude of crime
incidents in Chicago. We examined the burglaries occurring between
2005 and 2009, and spatially aggregated them within the 2010 census 
block groups.  Using the number of households
in each block group from the 2010 census, we then calculated the
number of burglaries per household over the considered time period.
We may think of each resulting proportion as a noisy measurement of
the underlying probability of burglary occuring in a randomly chosen
household within the given census block, over the 2005 to 2009 time
period.  These proportions are displayed in Figure \ref{fig:rob}.

\begin{figure}[p]
\centering
\includegraphics[width=\textwidth]{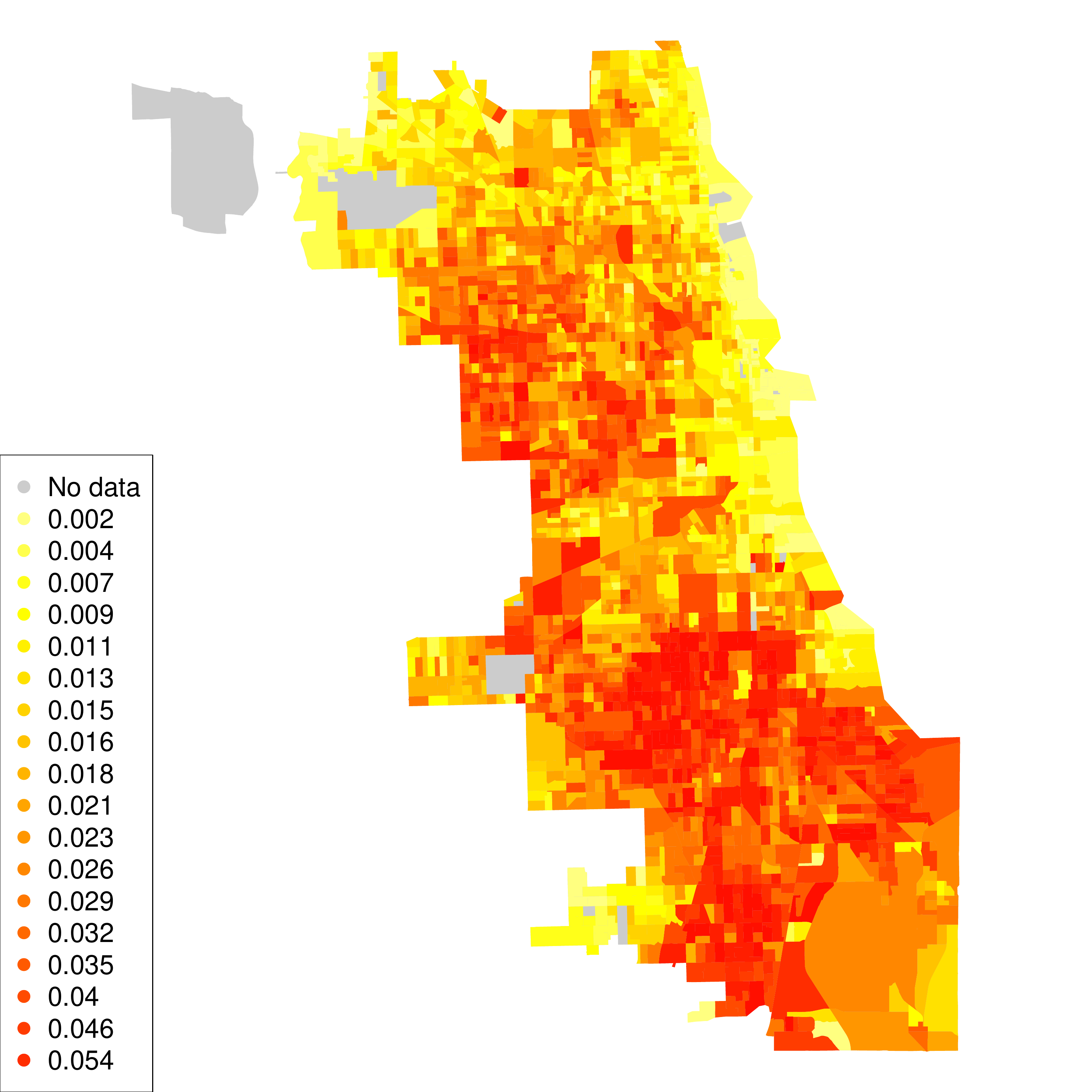}
\caption{\small\it Observed proportions of reported burglaries per
  household between 2005--2009 in Chicago, IL.  Data were aggregated
  within the 2010 census block groups.}
\label{fig:rob}
\end{figure}

\begin{figure}[p]
\centering
\includegraphics[width=\textwidth]{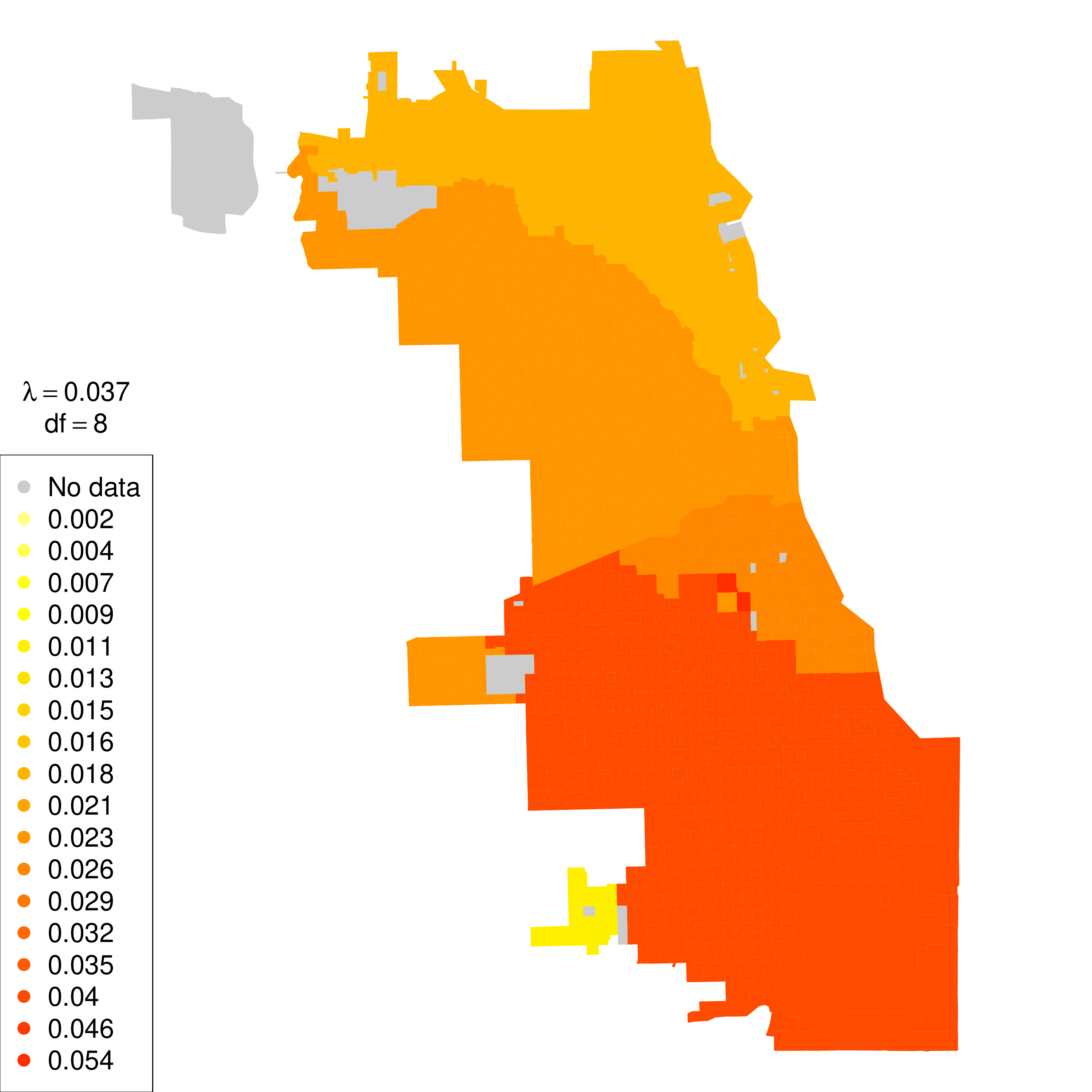}
\caption{\small\it A solution, corresponding to $\lambda = 0.037$,
  along the graph fused lasso path that was fit to the observed
  proportions of burglaries.}
\label{fig:rob10}
\end{figure}

We consider the task of estimating burglarly probabilities across
Chicago census blocks, and simultaneously grouping or clustering
these estimates across adjacent census blocks.  The fused lasso, with
an $\ell_1$ penalty on the differences between neighboring blocks,
provides a means of carrying out this task.  Using 
an $\ell_2$ or Huber penalty on the block differences would be
easier for optimization, but would not be appropriate for the goal at
hand because these smooth penalties are not capable of producing exact 
fusions in the components of the estimate.  The fused lasso setup for
the Chicago crime data used $n=2167$ blocks in total (nodes in the 
underlying graph), and $m=14,060$ connections between neighboring
blocks (edges in the graph).  Setting $X=I$, we computed the first
2500 steps of the fused lasso path, using the specialized
implementation of Section \ref{sec:fusedlasso}, which took a little
over a minute on a laptop computer.  The largest degrees of freedom
achieved by a solution in these first 2500 steps was 34.  

Figure \ref{fig:rob10} displays one particular fused lasso solution
from this path, corresponding to 8 degrees of freedom (Appendix
\ref{app:crime} displays other solutions). Note
that this solution divides the city into roughly four regions, with the
most risky region being the southern side of the city, and the least
risky being the northern side. In addition to the four main regions,
we can also see that a small region of the city with very low burglary
risk scores is isolated in the lower left part of the city; since it is
buffered by a corridor of census blocks with no data, the region
incurs only a small penalty for breaking off from the main graph. 
This picture in Figure \ref{fig:rob10} offers a better
qualitative understanding of large scale spatial patterns than do the
raw data in Figure \ref{fig:rob}.  It also provides a high level
clustering of census blocks which could be useful for police
dispatchers, city planners, politicians, and insurance companies.

% In addition to the four main regions, several other artefacts present
% themselves in this solution. 
% %The dark red disconnected region in the
% %upper left side of the plot is the Chicago O'Hare airport, cut off
% %from the rest of the city by local politics \citep{zornchi}; as
% %no fusion penalty is associated with it, its estimated proportion is
% %equal to the observed proportion, for all values of $\lambda$.  
% A small region of the city with very low burglary risk
% scores is isolated in the lower left of the city; since it is buffered
% by a corridor of census blocks with no data, the region incurs only a
% small penalty for breaking off from the main graph. Also, four
% singleton census blocks exist which are also broken from the rest of
% the graph. In these cases, the observed values are significantly
% outside the values of the neighboring blocks, generally due to a low
% number of households making the risk score more volatile. Singletones
% will often occur the edge of the graph where the incurred fused
% penalty is lower, or at the boundary of major clusters where there
% will be a large fused penalty, regardless of the estimate. 

Lastly, we remark that one benefit of the fused lasso over many
competing graph 
clustering methods is its local adaptivity. Simply put, the algorithm
will adaptively determine the size of a cluster given the nodewise 
measurements, counter to the tendency of other methods in
creating roughly equal sized clusters. 

\section{Empirical timings}
\label{sec:timings}

Table \ref{tab:runtimes} presented the theoretical per-iteration
complexities of the various specialized
implementations of the generalized lasso path algorithm.   Here we
briefly explore the empirical scaling of our implementations.
For many classes of generalized lasso problems, as the problem size
$n$ grows, the number of iterations taken by the
path algorithm before termination can increase super-linearly in $n$.
(A notable exception is the 1d fused lasso problem with $X=I$, in
which the number of iterations before termination is always $n-1$.)
For large problems, therefore, solving the entire path becomes
computationally infeasible, and also often undesirable (typically,
applications call for the more regularized solutions visited toward
the start of the  path).  Hence, we investigate the
time required to compute the first $100$ iterations of the path
algorithm; continuing further down the path should scale accordingly
with the number of steps.

\begin{figure}[htb]
\centering
\includegraphics[width=\textwidth]{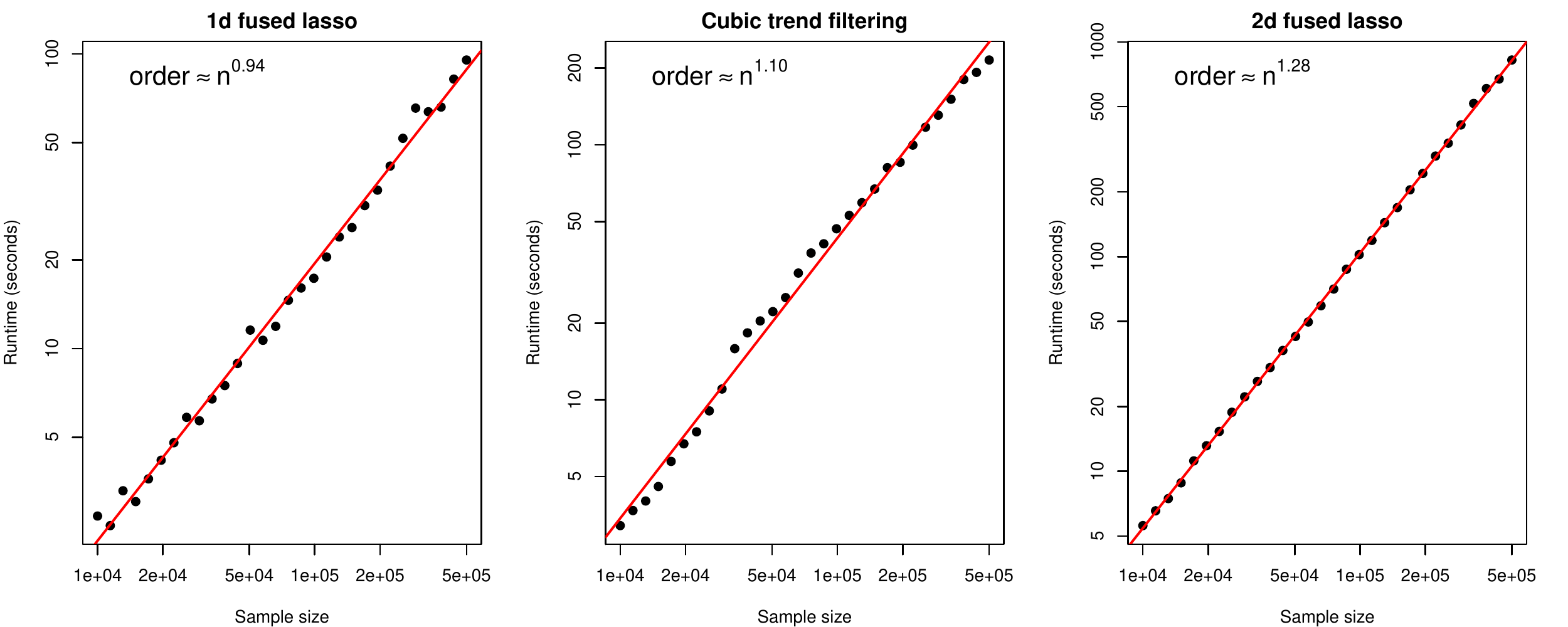}
\caption{\small\it Runtimes from computing the first 100
  steps of the generalized lasso path for 30 problem sizes ranging
  from $n=1000$ to $n=50,000$.  The left panel shows the results for
  1d fused lasso problems, the middle shows cubic trend filtering
  problems, and the right shows 2d fused lasso problems.  Each plot is
  on a log-log scale, and a least squares line (passing through 0) was
  fit to determine the empirical scaling of each implementation with
  $n$.}
\label{fig:time}
\end{figure}

The runtimes for the first $100$ path steps, with the
sample size $n$ varying from $1000$ to $50,000$, are presented
in Figure \ref{fig:time}. These were timed on a laptop computer.  We
considered three problem
classes, all with $X=I$: the 1d fused lasso, cubic trend filtering,
and the 2d fused lasso problem classes.  For the first two settings,
we generated noisy observations around a mean following a two-period
sinusoidal function.  For the 2d fused lasso setting, we generated
noisy observations over an approximately square grid, around a mean
that was elevated in the bottom third quadrant of the grid.  Note
that the empirical complexity of the first 100 steps, in both the 1d
fused lasso and trend filtering settings, is approximately linear, as
predicted by the theoretical analysis.  The steps in the 2d fused
lasso computation scales slightly slower, but still not far from
linear.

% The 2D fused lasso runs slightly slower than linear scaling, $\mathcal{O}(n^{1.28})$,
% which seems reasonable as the sparse matrix implementation picks up additional
% communication costs due to the fact that $D$ is not strictly diagonal.

Across all three settings, our empirically derived scalings indicate
that 100 path steps can be computed for problem sizes into the
millions within a relatively short (i.e., less than one hour) time
period.  This bodes
well for the 1d fused lasso and trend filtering problems, because
in these cases, hundreds or thousands of path steps can often deliver
regularized solutions of interest, even in very large problem sizes.
However, for the 2d fused lasso problem, it is more often the case
that many, many steps are needed to deliver solutions of
interest.  This has to do the connectivity of the graph
corresponding to $D_{-\cB}$, with $\cB$ being the boundary set---see
Section \ref{sec:lapfl}, or \citet{genlasso}.  We have found that the
number of steps needed for interesting solutions scales
more favorably when running the fused lasso on a graph determined by
geographic regions (e.g., census block groups, as in
Sections \ref{sec:crime}), but the number of steps
grows prohibitively large for grid graphs, especially in a setting
like image denoising, where the desired solutions often display a
large number of connected components and hence require many steps. 

Finally, we note that these runtimes were calculated using a
default version of R (specifically, R version 3.1).   As our
specialized implementations all use built-in R matrix functions in
one way or another, compiling R against a commercial matrix
library will likely improve these results drastically on multicore
machines.

\section{Discussion}
\label{sec:discuss}

We have developed efficient implementations of the generalized lasso
dual path algorithm of \citet{genlasso}.  In particular, we derived an
implementation for
a general penalty matrix $D$, one for trend filtering problems, in
which $D$ is the discrete difference operator of a given order, and
one for fused lasso problems, in which $D$ is the oriented incidence
matrix of some underlying graph.  Each implementation can handle the
signal approximator case, $X=I$, as well as a general predictor matrix
$X$.  These implementations are all put to use in the {\tt genlasso} R
package.

% Some possibly interesting related topics, for future work: a
% specialized implementation for sparse trend filtering problems (in
% which there is an additional \smash{$\ell_1$} penalty, next to the
% usual penalty on higher order differences between neighboring
% coefficients), a specialized implementation for varying coefficient
% models fit using trend filtering, and ``backwards'' path algorithms
% that begin at $\lambda=0$ and solve the path as $\lambda$ increases
% until $\lambda=\infty$.

\section*{Acknowledgements}

RT was supported by NSF Grant DMS-1309174.

\appendix

\section{The QR decomposition and least squares problems}
\label{app:qrls}

Here we give a brief review of the QR decomposition, and the application of
this decomposition to least squares problems. Chapter 5 of \citet{gvl}
is an excellent reference.

\subsection{The QR decomposition of a full column rank matrix}
\label{app:qrfull}

Let $A \in \R^{m \times n}$ with $\rank(A)=n$ (this implies that $m \geq n$).
Then there exists matrices $Q \in \R^{m\times m}$ and $R\in \R^{m\times n}$
such that $A=QR$, where $Q$ is orthogonal (its first $n$ columns
form a basis for the column space of $A$), and $R$ is of the form
\begin{equation*}
R = \left[\begin{array}{c} R_1 \\ 0 \end{array}\right],
\end{equation*}
$R_1 \in \R^{n\times n}$ being upper triangular.
This is (not surprisingly) called the {\it QR decomposition}, and it
can be computed in $O(mn^2)$ operations \citep{gvl}.

The decomposition $A=QR$ is used primarily for solving least squares problems.
For example, given $b \in \R^m$, suppose that are interested in finding
$x \in \R^n$ to minimize
\begin{equation}
\label{eq:ls}
\|b - Ax\|_2^2.
\end{equation}
Since $\rank(A)=n$, the minimizer $x$---also referred to as the solution---is
unique. Let $Q_1 \in \R^{m\times n}$ denote the first $n$ columns of $Q$,
and let $Q_2 \in \R^{m\times(m-n)}$ denote the last $m-n$ columns. Then
\begin{equation*}
\|b-Ax\|_2^2 = \|Q^T(b-Ax)\|_2^2 = \|Q_1^T b - R_1x\|_2^2 + \|Q_2^T b\|_2^2,
\end{equation*}
and so minimizing the left-hand side is equivalent to minimizing
$\|Q_1^T b - R_1x\|_2^2$. This can be done quickly, by solving the equation
$R_1x = c$ where $c=Q_1^T b$. Recalling the triangular structure of $R_1$,
this looks like:
\begin{equation*}
\sw\sh
\left[\begin{array}{ccccc}
\Box & \Box & \Box & \Box & \Box \\
& \Box & \Box & \Box & \Box \\
& & \Box & \Box & \Box \\
& & & \Box & \Box \\
& & & & \Box
\end{array}\right] \cdot x = c,
\end{equation*}
where the boxes denote nonzero entries, and blank spaces indicate zero
entries. We first solve the equation given by last row (an equation in one
variable), then we substitute and solve the second to last row, etc. This
{\it back-solve} procedure takes $O(n^2)$ operations. Hence, finding the least
squares solution of \eqref{eq:ls} requires $O(mn)+O(n^2)=O(mn)$ operations
in total (the first term counts the multiplication by $Q_1^T$ to form $c=Q_1^T b$).
Note that this does not count the $O(mn^2)$ operations required to compute the
QR decomposition of $A$ in the first place; and importantly, if we want to
minimize multiple criterions of the form \eqref{eq:ls} for different vectors
$b$, then we only compute the QR decomposition of $A$ once, and use this
decomposition to find each solution quickly in $O(mn)$ operations.

\subsection{The QR decomposition of a column rank deficient matrix}
\label{app:qrdef}

Let $A \in \R^{m\times n}$ with $\rank(A)=k \leq n$. Then there exists
$P \in \R^{n\times n}$, $Q \in \R^{m\times m}$, and
$R \in \R^{m \times n}$ such that $AP=QR$, where $P$ is a permutation matrix,
$Q$ is orthogonal (its first $k$ columns span the column space of $A$),
and $R$ decomposes as
\begin{equation*}
R=\left[\begin{array}{cc}
R_1 & R_2 \\
0 & 0
\end{array}\right],
\end{equation*}
where $R_1 \in \R^{k\times k}$ is upper triangular, and
$R_2 \in \R^{k \times (n-k)}$ is dense. Visually, $R$
looks like this (when the order of rank deficiency is $n-k=2$):
\begin{equation*}
\sw\sh
\left[\begin{array}{ccccc}
\Box & \Box & \Box & \Box & \Box \\
& \Box & \Box & \Box & \Box \\
& & \Box & \Box & \Box \\
& & & \wBox & \wBox \\
& & & & \wBox
\end{array}\right].
\end{equation*}
Note that $AP$ just permutes the columns of $A$.
This decomposition takes $O(mnk)$ operations \citep{gvl}.

The least squares criterion in \eqref{eq:ls} can now admit many
solutions $x$ (in fact, infinitely many) if $\rank(A)<n$.
If we simply want any solution $x$---\citet{gvl} refer to this as a
{\it basic solution}---then we can use the QR decomposition $AP=QR$.
We write
\begin{align*}
\|b-Ax\|_2^2 &= \|b-APP^T x\|_2^2 \\
&= \|Q^T(b-APP^T x)\|_2^2 \\
&=\|Q_1^T b - \big[ R_1 \;\, R_2 \big]
P^T x\|_2^2 + \|Q_2^T b\|_2^2,
\end{align*}
where $Q_1\in\R^{m\times k}$ contains the first $k$ columns of $Q$,
and $Q_2\in\R^{m\times (m-k)}$ contains the last $m-k$ columns.
We can now consider $z=P^T x$ as the optimization variable, and solve
\begin{equation}
\label{eq:req}
\Big[\begin{array}{cc}
R_1 & R_2
\end{array}\Big]
\left[\begin{array}{c}
z_1 \\ z_2
\end{array}\right]
= Q_1^T b,
\end{equation}
where we have decomposed $z=(z_1,z_2)$ with $z_1 \in \R^{k}$ and
$z_2 \in \R^{n-k}$. Note that to solve
\eqref{eq:req}, we can take $z_2=0$, and then back-solve to compute $z_1$ in
$O(k^2)$ operations. Letting $x=Pz$, we have hence computed a basic least
squares solution in $O(mk)+O(k^2)+O(n)=O(mn)$ operations.

\subsection{The minimum $\ell_2$ norm least squares solution}
\label{app:mls}

Suppose again that $A\in\R^{m\times n}$ and $\rank(A)=k \leq n$. If we
want to compute the unique solution\footnote{Uniqueness follows from the fact that
the set of least squares solutions forms a convex set. Note that this is given
by $x^*=A^+b$, where $A^+$ is the Moore-Penrose pseudoinverse of $A$.}
%This pseudoinverse can be defined in terms of the singular value decomposition of $A$:
%if we write $A=UDV^T$, where $k=\rank(A)$, $U \in \R^{m\times k}, V \in \R^{n\times k}$
%have orthonormal columns, and $D \in \R^{k\times k}$ is diagonal, then
%$A^+=VD^{-1} U^T$. Hence if $u_i \in \R^m$ denotes the $i$th column of $U$ and
%$v_i \in \R^n$ the $i$th column of $V$,
%then $x^*=A^+b=\sum_{i=1}^k u_i v_i^T b / d_i$.
$x^*$ that has the minimum $\ell_2$ norm across all least squares solutions
$x$ in \eqref{eq:ls}, then the strategy given in the last section does not
necessarily work (in fact, it does not produce $x^*$ unless $R_2=0$).
However, we can modify the QR decomposition $AP=QR$ from Section
\ref{app:qrdef} in order to compute $x^*$. For this, we need to apply
Givens rotations to $R$.
These are covered in the next section, but for now, the key message is
that there exists an orthogonal transformation $G \in \R^{n\times n}$ such that
\begin{equation}
\label{eq:tilr}
RG = \tR = \left[\begin{array}{cc}
0 & \tR_1 \\
0 & 0
\end{array}\right],
\end{equation}
where \smash{$\tR_1 \in \R^{k\times k}$} is upper triangular. Applying
a single Givens rotation to (the columns of) $R$ takes $O(k)$
operations, and $G$ is composed
of $k(n-k)$ of them, so forming $RG$ takes $O(k^2(n-k))$
operations. Hence the decomposition \smash{$APG = Q\tR$} requires the
same order of complexity, $O(mnk)+O(k^2(n-k))=O(mnk)$ operations in
total.

Now we write
\begin{equation*}
\|b-Ax\|_2^2 = \|b-APGz\|_2^2,
\end{equation*}
where $z=G^T P^T x$. Since $P,G$ are orthogonal, we have
$\|x\|_2=\|z\|_2$,
and therefore our problem is equivalent to finding the
minimum $\ell_2$ norm minimizer $z^*$ of the right-hand side above.
As before, we now utilize the QR decomposition, writing
\begin{equation*}
\|b-APGz\|_2^2 = \|Q^T(b-APGz)\|_2^2
= \|Q_1^T b - \big[0 \;\, \tR_1\big] z\|_2^2 + \|Q_2^T b\|_2^2,
\end{equation*}
where $Q_1\in\R^{m\times k}$ and $Q_2\in\R^{m\times(m-k)}$ give,
respectively, the first $k$ and the last $m-k$ columns of $Q$.
Hence we seek the minimum $\ell_2$ norm solution of
\begin{equation*}
\Big[\begin{array}{cc}
0 & \tR_1
\end{array}\Big]
\left[\begin{array}{c}
z_1 \\ z_2
\end{array}\right]
= Q_1^T b,
\end{equation*}
where $z=(z_1,z_2)$ with $z_1 \in \R^{(n-k)}$ and $z_2 \in \R^k$.
For $z^*$ to have minimum $\ell_2$ norm, we must have $z^*_1=0$. Then
$z^*_2$ is given by back-solving, which takes $O(k^2)$ operations.
Finally, we let $x^*=PGz^*$, and count $O(mk)+O(k^2)+O(n^2)+O(n)
= O(n\cdot\max\{m,n\})$ operations in total to compute the minimum
$\ell_2$ norm least squares solution.

\subsection{The minimum $\ell_2$ norm least squares solution and
the transposed QR}
\label{app:mlst}

Given $A\in\R^{m\times n}$ with $\rank(A)=k \leq n$, it can be advantageous
in some problems to use a QR decomposition of $A^T$ instead of $A$.
(For example, this is the case when we want to update the QR decomposition
after $A$ has changed by one column; see
Section \ref{app:updefcol}.) By what we just showed,
we can compute a decomposition \smash{$A^T PG = Q\tR$}, where
$P \in \R^{m\times m}$ is a permutation matrix, $G \in \R^{m\times m}$
is an orthogonal matrix of Givens rotations, $Q \in \R^{n\times n}$ is
orthogonal, and \smash{$\tR \in \R^{n\times m}$} is of the special
form \eqref{eq:tilr} with \smash{$\tR_1 \in \R^{k\times k}$} upper
triangular, in $O(mnk)$ operations.

To find the minimum $\ell_2$ norm minimizer $x^*$ of \eqref{eq:ls},
we can employ a similar strategy to that of Section \ref{app:mls}.
Using the orthogonality of $P,G$, and the computed decomposition
\smash{$G^T P^T A = \tR^T Q^T$}, we have
\begin{equation*}
\|b-Ax\|_2^2 = \|G^T P^T(b-Ax)\|_2^2 =
\|c_1 - \big[\tR_1^T \;\, 0\big] z\|_2^2 + \|c_2\|_2^2,
\end{equation*}
where $z=Q^T x$, and $c_1,c_2$ denote the first $k$, respectively
last $m-k$ coordinates of $c=G^T P^T b$.
As $Q$ is orthogonal, we have $\|x\|_2=\|z\|_2$, and
hence it suffices to find the minimum $\ell_2$ norm minimizer $z^*$ of
the right-hand side above. This is the same as finding the minimum
$\ell_2$ norm solution of the linear equation
\begin{equation*}
\Big[\begin{array}{cc}
\tR_1^T & 0
\end{array}\Big]
\left[\begin{array}{c}
z_1 \\ z_2
\end{array}\right]
=  c_2,
\end{equation*}
where $z=(z_1,z_2)$ with $z_1 \in \R^k$ and $z_2 \in \R^{(n-k)}$.
Therefore $z_2^*=0$, and $z_1^*$ can be computed by foward-solving (the
same concept as back-solving, except we start with the first row), requiring
$O(k^2)$ operations. We finally take $x^*=Qz^*$. The total number
of operations is $O(n)+O(m^2)+O(k^2)+O(nk)=O(m\cdot\max\{m,n\})$.

\section{Givens rotations}
\label{app:givens}

We describe Givens rotations, orthogonal transformations that help
maintain (or create) maintain upper triangular structure. Givens
rotations provide a way to efficiently update the QR decomposition
of a given matrix after a row or column has been added or deleted. (They
also provide a way to compute the QR decomposition in the first place.)
Our explanation and notation here are based largely on Chapter 5
of \citet{gvl}.

\subsection{Simple Givens rotations in two dimensions}
\label{app:givens2d}

The main idea behind a Givens rotation can be expressed by considering
the $2 \times 2$ rotation matrix
\begin{equation*}
G =
\left[\begin{array}{cc}
c & s \\ -s & c
\end{array}\right],
\end{equation*}
where $c=\cos\theta$ and $s=\sin\theta$, for some $\theta\in[0,2\pi]$.
Multiplication by $G^T$ amounts to a counterclockwise rotation through
an angle $\theta$; since it is a rotation matrix, $G$ is clearly orthogonal.
Furthermore, given any vector $(a,b) \in \R^2$, we can choose $c,s$
(choose $\theta$) such that
\begin{equation*}
G^T
\left[\begin{array}{c}
a \\ b
\end{array}\right] =
\left[\begin{array}{c}
d \\ 0
\end{array}\right],
\end{equation*}
for some $d\in\R$.
This is simply rotating $(a,b)$ onto the first coordinate axis, and
by inspection we see that we must take $c=a/\sqrt{a^2+b^2}$ and
$s=-b/\sqrt{a^2+b^2}$. Note that, from the point of view of computational
efficiency, we never have to compute $\theta$ (which would require
inverse trigonometric functions).

\subsection{Givens rotations in higher dimensions}
\label{app:givenshd}

The same idea extends naturally to higher dimensions. Consider the
$n\times n$ Givens rotation matrix
\begin{equation*}
G = \left[\begin{array}{cccccccc}
1 & 0 & \ldots & 0 & \ldots & 0 & \ldots & 0 \\
0 & 1 & \ldots & 0 & \ldots & 0 & \ldots & 0 \\
\vdots & & & & & & & \\
0 & 0 & \ldots & c & \ldots & s & \ldots & 0 \\
\vdots & & & & & & & \\
0 & 0 & \ldots & -s & \ldots & c & \ldots & 0 \\
\vdots & & & & & & & \\
0 & 0 & \ldots & 0 & \ldots & 0 & \ldots & 1 \\
\end{array}\right];
\end{equation*}
in other words, $G$ is the $n\times n$ identity matrix, except with
four elements $G_{ii},G_{ij},G_{ji},G_{jj}$ replaced with the
corresponding elements of the $2\times 2$ Givens rotation matrix.
We will write $G=G(i,j)$ to emphasize the dependence on $i,j$.
It is straightforward to check that $G$ is orthogonal. Applying $G^T$ to
a vector $x\in\R^n$ only affects components $x_i$ and $x_j$, and leaves all
other components untouched: with $z=G^T x$, we have
\begin{equation*}
z_k = \begin{cases}
cx_i - sx_j & \text{if}\;\, k=i \\
sx_i + cx_j & \text{if}\;\, k=j \\
x_k & \text{otherwise}
\end{cases}.
\end{equation*}
Because $G^T$ only acts on two components, we can compute
$z=G^T x$ in $O(1)$ operations. And as in the $2\times 2$ case, we can make
$z_j=0$ by taking
\begin{equation*}
c=x_i/\sqrt{x_i^2+x_j^2}\;\;\; \text{and} \;\;\;
s=-x_j/\sqrt{x_i^2+x_j^2}.
\end{equation*}

Now we consider Givens rotations applied to matrices. If $A\in\R^{m\times n}$
and $G=G(i,j)\in\R^{m\times m}$, then pre-multiplying $A$ by $G^T$ (as in $G^T A$)
only affects rows $i$ and $j$, and hence computing $G^T A$ takes $O(n)$
operations. Moreover, with the appropriate choice of $c,s$, we can selectively
zero out an element in the $j$th row of $G^T A$. A common application of $G^T$
looks like the following:
\begin{equation*}
\sw\sh
G^T \cdot
\left[\begin{array}{cccc}
\Box & \Box & \Box & \Box \\
& \Box & \Box & \Box \\
& & \Box & \Box \\
& & \Box & \Box \\
& & & \Box
\end{array}\right] =
\left[\begin{array}{cccc}
\Box & \Box & \Box & \Box \\
& \Box & \Box & \Box \\
& & \Box & \Box \\
& & & \Box \\
& & & \Box
\end{array}\right],
\end{equation*}
where in this example $G=G(3,4)$, and $c,s$ have been chosen so that the
element in the 4th row and 3rd column of the output is zero. Importantly,
the first 2 columns of rows 3 and 4 were all zeros to begin with, and zeros
after pre-multiplication, so that this zero pattern has not been disturbed
(think of the $2\times 2$ case: rotating $(0,0)$ still gives $(0,0)$).
Applying a second Givens rotation to the output gives an upper
triangular structure:
\begin{equation*}
\sw\sh
G_2^T \cdot
\left[\begin{array}{cccc}
\Box & \Box & \Box & \Box \\
& \Box & \Box & \Box \\
& & \Box & \Box \\
& & & \Box \\
& & & \Box
\end{array}\right] =
\left[\begin{array}{cccc}
\Box & \Box & \Box & \Box \\
& \Box & \Box & \Box \\
& & \Box & \Box \\
& & & \Box \\
& & & \wBox
\end{array}\right],
\end{equation*}
where $G_2=G_2(4,5)$ is another Givens rotation
matrix.

On the other hand, post-multiplying $A\in\R^{m\times n}$ by a Givens
rotation matrix $G=G(i,j)\in\R^{n\times n}$ (as in $AG$) only affects
columns $i$ and $j$. Therefore computing $AG$ requires $O(m)$ operations.
The logic is very similar to the pre-multiplication case, and by choosing
$c,s$ appropriately, we can zero out a particular element in the $j$th
column of $AG$. A common application looks like:
\begin{equation*}
\sw\sh
\left[\begin{array}{ccccc}
\Box & \Box & \Box & \Box & \Box \\
& \Box & \Box & \Box & \Box \\
& & \Box & \Box & \Box \\
& & & & \Box
\end{array}\right] \cdot G =
\left[\begin{array}{ccccc}
\Box & \Box & \Box & \Box & \Box \\
& \Box & \Box & \Box & \Box \\
& & & \Box & \Box \\
& & & & \Box
\end{array}\right],
\end{equation*}
where $G=G(3,4)$ and $c,s$ were chosen to zero out the element in the
3rd row and 3rd column. Now applying two more Givens rotations yields
an upper triangular structure:
\begin{equation*}
\sw\sh
\left[\begin{array}{ccccc}
\Box & \Box & \Box & \Box & \Box \\
& \Box & \Box & \Box & \Box \\
& & & \Box & \Box \\
& & & & \Box
\end{array}\right] \cdot G_2 G_3 =
\left[\begin{array}{ccccc}
\wBox & \Box & \Box & \Box & \Box \\
& & \Box & \Box & \Box \\
& & & \Box & \Box \\
& & & & \Box
\end{array}\right].
\end{equation*}

\section{Updating the QR decomposition in the full rank case}
\label{app:upfull}

In this section we cover techniques based on Givens rotations for
updating the QR decomposition of a matrix $A \in \R^{m\times n}$,
after a row or column has been either added or removed to $A$.
We assume here that $\rank(A)=n$; the next section covers the rank
deficient case, which is more delicate. For the full rank update
problem, a good reference is Section 12.5 of \citet{gvl}.
Hence suppose that we have computed a decomposition $A=QR$,
with $Q\in\R^{m\times m}$ and $R\in\R^{m\times n}$, as described
in Section \ref{app:qrfull}, and
we subsequently want to compute a QR decomposition of \smash{$\tA$},
where \smash{$\tA$} differs from $A$ by either one row or one
column. As motivation, we may have already solved the least squares
problem
\begin{equation*}
\|b-Ax\|_2^2,
\end{equation*}
and now want to solve the new least squares problem
\begin{equation*}
\|c-\tA x\|_2^2.
\end{equation*}
As we will see, computing a QR decomposition of \smash{$\tA$} by
updating that of $A$ saves an order of magnitude in computational time
when compared to the naive route (computing the QR decomposition
``from scratch''). We treat the row and column update problems
separately.

\subsection{Adding or removing a row}
\label{app:upfullrow}

Suppose that \smash{$\tA\in\R^{(m+1)\times n}$} is formed by adding a
row to $A$, following its $i$th row, so
\begin{equation*}
A = \left[\begin{array}{c}
A_1 \\ A_2
\end{array}\right]
\;\;\;\text{and}\;\;\;
\tA = \left[\begin{array}{c}
A_1 \\ w^T \\ A_2 \end{array}\right],
\end{equation*}
where $A_1 \in \R^{i\times n}$, $A_2\in\R^{(m-i)\times n}$, and $w\in\R^n$
is the row to be added.
Let $Q_1\in\R^{i\times m}$ denote the first $i$ rows of $Q$ and $Q_2 \in
\R^{(m-i)\times m}$ denote its last $m-i$ rows. By rearranging both the rows
of $A$ the rows of $Q$ in the same way, the product $Q^T A$ remains the same:
\begin{equation*}
\left[\begin{array}{cc} Q_2^T & Q_1^T \end{array}\right]
\left[\begin{array}{c} A_2 \\ A_1 \end{array}\right] =
Q^T A = R
= \left[\begin{array}{c} R_1 \\ 0 \end{array}\right],
\end{equation*}
where $R_1\in \R^{n\times n}$ is upper triangular. Therefore
\begin{equation*}
\left[\begin{array}{ccc}
1 & 0 & 0 \\
0 & Q_2^T & Q_1^T
\end{array}\right]
\left[\begin{array}{c}
w^T \\ A_2 \\ A_1
\end{array}\right] =
\left[\begin{array}{c}
w^T \\ R_1 \\ 0
\end{array}\right].
\end{equation*}
We can now apply Givens rotations $G_1,\ldots G_n$ so that
\begin{equation*}
\tR = G_n^T \ldots G_1^T
\left[\begin{array}{c}
w^T \\ R_1 \\ 0
\end{array}\right]
= \left[\begin{array}{c}
\tR_1 \\ 0
\end{array}\right],
\end{equation*}
where \smash{$\tR_1\in\R^{n\times n}$} is upper triangular. Hence
defining
\begin{equation*}
\tQ = \left[\begin{array}{cc}
0 & Q_1 \\
1 & 0 \\
0 & Q_2
\end{array}\right]
G_1\ldots G_n,
\end{equation*}
and noting that \smash{$\tQ\in\R^{(m+1)\times(m+1)}$} is still
orthogonal, we have \smash{$\tA=\tQ\tR$}, the desired QR
decomposition. This update
procedure uses $n$ Givens rotations, and therefore it requires a total
of $O(mn)$ operations. Compare this to the usual cost $O(mn^2)$ of
computing a QR decomposition (without updating).

On the other hand, suppose that \smash{$\tA \in \R^{(m-1)\times n}$}
is formed by removing the $i$th row of $A$. Hence
\begin{equation*}
A = \left[\begin{array}{c}
A_1 \\ w^T \\ A_2 \end{array}\right]
\;\;\;\text{and}\;\;\;
\tA = \left[\begin{array}{c}
A_1 \\ A_2
\end{array}\right],
\end{equation*}
where $A_1 \in \R^{(i-1)\times n}$, $A_2\in\R^{(m-i)\times n}$, and $w\in\R^n$
is the row to be deleted.  (We assume without a loss of generality
that $m>n$, so that removing a row does
not change the rank; updates in the rank deficient case are covered in
the next section.)
Let $q\in\R^m$ denote the $i$th row of $Q$, and note that we can
compute Givens rotations $G_1,\ldots G_{m-1}$ such that
\begin{equation*}
G_{m-1}^T \ldots G_1^T q = se_1,
\end{equation*}
with $e_1=(1,0,\ldots 0) \in\R^m$ the first standard basis vector,
and $s=\pm 1$. Let \smash{$\tQ=Q G_1\ldots G_{m-1}$}; then, as
\smash{$\tQ$} is still orthogonal, it has the form
\begin{equation*}
\tQ = \left[\begin{array}{cc}
0 & \tQ_1 \\
s & 0 \\
0 & \tQ_2
\end{array}\right],
\end{equation*}
where \smash{$\tQ_1\in\R^{(i-1)\times (m-1)}$} and
\smash{$\tQ_2\in\R^{(m-i)\times (m-1)}$}. Furthermore, defining
\smash{$\tR G_{m-1}^T \ldots G_1^T R$}, we can see that
\begin{equation*}
\tR =
\left[\begin{array}{c}
v^T \\ \tR_1 \\ 0
\end{array}\right],
\end{equation*}
where \smash{$\tR_1 \in \R^{n\times n}$} is upper triangular and
$v\in\R^n$. By construction \smash{$A = QR = \tQ\tR$}, and defining
\begin{equation*}
\tQ_0 =
\left[\begin{array}{c}
\tQ_1 \\ \tQ_2
\end{array}\right]
\;\;\;\text{and}\;\;\;
\tR_0 = \left[\begin{array}{c} \tR_1 \\ 0 \end{array}\right],
\end{equation*}
we have \smash{$\tA=\tQ_0\tR_0$}, as desired. We performed $m-1$
Givens rotations, and hence $O(m^2)$ operations.

\subsection{Adding or removing a column}
\label{app:upfullcol}

Suppose that \smash{$\tA\in\R^{m\times(n+1)}$} is formed by adding a
column to $A$, say, after its $j$th column. Then
\begin{equation*}
Q^T \tA = \left[\begin{array}{ccc}
R_1 & \vdots & R_2 \\
0 & w & R_3 \vspace{-5pt} \\
0 & \vdots & 0
\end{array}\right],
\end{equation*}
where $R_1\in\R^{j\times j}$ and $R_3\in\R^{(n-j)\times(n-j)}$
are upper triangular, $R_2\in\R^{j\times(n-j)}$ is dense, and
$w\in\R^m$.
(We are assuming here, without a loss of generality, that the added column
does not lie in the span of the existing ones; updates in the rank
deficient case are covered in the next section.)
%% As an example, inserting a column after the 2nd column (when $m=6$ and
%% $n=5$) gives
%% \begin{equation*}
%% Q^T \tA =
%% \sw\sh
%% \left[\begin{array}{cccccc}
%% \Box & \Box & \Box & \Box & \Box & \Box \\
%% & \Box & \Box & \Box & \Box & \Box \\
%% & & \Box & \Box & \Box & \Box \\
%% & & \Box & & \Box & \Box \\
%% & & \Box & & & \Box \\
%% & & \Box & & & \wBox
%% \end{array}\right].
%% \end{equation*}
We can apply Givens rotations $G_1,\ldots G_{n-j}$ to the rows of
\smash{$Q^T \tA$} so that
\begin{equation*}
G_{n-j}^T \ldots G_1^T Q^T \tA =
\left[\begin{array}{c} \tR_1 \\ 0 \end{array}\right] = \tR,
\end{equation*}
where \smash{$\tR_1\in\R^{(n+1)\times(n+1)}$} is upper triangular.
%% In terms of our example, after the first Givens rotation we get
%% \begin{equation*}
%% G_1^T \cdot
%% \sw\sh
%% \left[\begin{array}{cccccc}
%% \Box & \Box & \Box & \Box & \Box & \Box \\
%% & \Box & \Box & \Box & \Box & \Box \\
%% & & \Box & \Box & \Box & \Box \\
%% & & \Box & & \Box & \Box \\
%% & & \Box & & & \Box \\
%% & & \Box & & & \wBox
%% \end{array}\right] =
%% \left[\begin{array}{cccccc}
%% \Box & \Box & \Box & \Box & \Box & \Box \\
%% & \Box & \Box & \Box & \Box & \Box \\
%% & & \Box & \Box & \Box & \Box \\
%% & & \Box & & \Box & \Box \\
%% & & \Box & & & \Box \\
%% & & & & & \Box
%% \end{array}\right],
%% \end{equation*}
%% and after two more,
%% \begin{equation*}
%% G_3^T G_2^T G_1^T \cdot
%% \sw\sh
%% \left[\begin{array}{cccccc}
%% \Box & \Box & \Box & \Box & \Box & \Box \\
%% & \Box & \Box & \Box & \Box & \Box \\
%% & & \Box & \Box & \Box & \Box \\
%% & & \Box & & \Box & \Box \\
%% & & \Box & & & \Box \\
%% & & \Box & & & \wBox
%% \end{array}\right] =
%% \left[\begin{array}{cccccc}
%% \Box & \Box & \Box & \Box & \Box & \Box \\
%% & \Box & \Box & \Box & \Box & \Box \\
%% & & \Box & \Box & \Box & \Box \\
%% & & & \Box & \Box & \Box \\
%% & & & & \Box & \Box \\
%% & & & & & \Box
%% \end{array}\right].
%% \end{equation*}
Therefore with \smash{$\tQ=Q G_1 \ldots G_{n-j}$}, we have
\smash{$\tA=\tQ\tR$}. We applied $O(n)$ Givens rotations, so this
update procedure requires $O(mn)$ operations.

If instead \smash{$\tA \in \R^{m\times (n-1)}$} is
formed by removing the $j$th column of $A$, then
\begin{equation*}
Q^T \tA = \left[\begin{array}{cc}
R_1 & R_2 \\
0 & w^T \\
0 & R_3 \\
0 & 0
\end{array}\right],
\end{equation*}
where $R_1\in\R^{(j-1)\times(j-1)}$ and $R_3\in\R^{(n-j)\times(n-j)}$
are upper triangular, $R_2\in\R^{(j-1)\times(n-j)}$ is dense, and
$w\in\R^{n-j}$.
%% For example, removing the 2nd column (when $m=6$ and $n=5$) looks like
%% \begin{equation*}
%% Q^T \tA =
%% \sw\sh
%% \left[\begin{array}{cccc}
%% \Box & \Box & \Box & \Box \\
%% & \Box & \Box & \Box \\
%% & \Box & \Box & \Box \\
%% & & \Box & \Box \\
%% & & & \Box \\
%% & & & \wBox
%% \end{array}\right].
%% \end{equation*}
Note that we can apply Givens rotations $G_1,\ldots G_{n-j}$ to the
rows of \smash{$Q^T \tA$} to produce
\begin{equation*}
G_{n-j}^T \ldots G_1^T Q^T \tA =
\left[\begin{array}{cc}
\tR_1 \\ 0
\end{array}\right] = \tR,
\end{equation*}
where \smash{$\tR_1 \in \R^{(n-1)\times(n-1)}$} is upper triangular.
%% Continuing
%% with the above example, we can apply three Givens rotations to the rows of
%% $Q^T \tA$ in order to achieve an upper triangular form,
%% \begin{equation*}
%% \sw\sh
%% G_3^T G_2^T G_1 ^T \cdot
%% \left[\begin{array}{cccc}
%% \Box & \Box & \Box & \Box \\
%% & \Box & \Box & \Box \\
%% & \Box & \Box & \Box \\
%% & & \Box & \Box \\
%% & & & \Box \\
%% & & & \wBox
%% \end{array}\right] =
%% \left[\begin{array}{cccc}
%% \Box & \Box & \Box & \Box \\
%% & \Box & \Box & \Box \\
%% & & \Box & \Box \\
%% & & & \Box \\
%% & & & \wBox \\
%% & & & \wBox
%% \end{array}\right].
%% \end{equation*}
Hence with \smash{$\tQ=Q G_1 \ldots G_{n-j}$}, we see that
\smash{$\tA=\tQ\tR$}. Again we used $O(n)$ Givens rotations, and
$O(mn)$ operations.

\section{Updating the QR decomposition in the rank deficient case}
\label{app:updef}

Here we again consider techniques for updating a QR decomposition, but
study the more difficult case in which $A\in\R^{m\times n}$ with $\rank(A)=k\leq n$.
In particular, we are interested in computing the minimum $\ell_2$ norm minimizer
of
\begin{equation}
\label{eq:ls2}
\|b-Ax\|_2^2,
\end{equation}
and subsequently, computing the minimum $\ell_2$ norm minimizer of
\begin{equation}
\label{eq:ls3}
\|c-\tA x\|_2^2.
\end{equation}
where \smash{$\tA$} has either one more or one less row that $A$, or
else one more of one less column than $A$.
Depending on whether whether our goal is to update the rows or columns,
we actually need to use a different QR decomposition for the the initial
least squares problem \eqref{eq:ls2}.

\subsection{Adding or removing a row}
\label{app:updefrow}

We compute the minimum $\ell_2$ norm minimizer of the initial least squares
criterion \eqref{eq:ls2}
using the QR decomposition $APG = QR$ described in Section \ref{app:mls},
where $P\in\R^{n\times n}$ is a permutation matrix,
$G\in\R^{n\times n}$ is a product of Givens rotations matrices,
$Q\in\R^{m\times m}$ and $R\in\R^{m\times n}$ is of the special form
\begin{equation*}
R = \left[\begin{array}{cc} 0 & R_1 \\ 0 & 0 \end{array}\right],
\end{equation*}
with $R_1\in\R^{k\times k}$ upper triangular (we note, in order to avoid
confusion, that $R,R_1$ were written as \smash{$\tR,\tR_1$} in Section
\ref{app:mls}).

First suppose that \smash{$\tA\in\R^{(m+1)\times n}$} is formed by
adding a row to $A$, after its $i$th row. Write
\begin{equation*}
A = \left[\begin{array}{c}
A_1 \\ A_2
\end{array}\right]
\;\;\;\text{and}\;\;\;
\tA = \left[\begin{array}{c}
A_1 \\ w^T \\ A_2 \end{array}\right],
\end{equation*}
where $A_1 \in \R^{i\times n}$, $A_2\in\R^{(m-i)\times n}$, and $w\in\R^n$
is the row to be added. Also let $Q_1\in\R^{i\times m}$ and
$Q_2\in\R^{(m-i)\times m}$ denote the first $i$ and last $m-i$ rows of $Q$,
respectively. The logic at this step is similar to that in the full rank case:
we can rearrange both the rows of $A$ and the
rows of $Q$ so that the product $Q^T A$ will not change, hence
\begin{equation*}
\left[\begin{array}{cc} Q_2^T & Q_1^T \end{array}\right]
\left[\begin{array}{c} A_2 \\ A_1 \end{array}\right] PG =
Q^T A P G = R
= \left[\begin{array}{cc} 0 & R_1 \\ 0 & 0 \end{array}\right].
\end{equation*}
Therefore
\begin{equation}
\label{eq:tc}
\left[\begin{array}{ccc}
1 & 0 & 0 \\
0 & Q_2^T & Q_1^T
\end{array}\right]
\left[\begin{array}{c}
w^T \\ A_2 \\ A_1
\end{array}\right] PG =
\left[\begin{array}{c}
w^T PG \\ Q^T A P G
\end{array}\right] =
\left[\begin{array}{cc}
d_1^T & d_2^T \\ 0 & R_1 \\ 0 & 0
\end{array}\right],
\end{equation}
where $d_1\in\R^{n-k}$ and $d_2\in\R^k$ are the first $n-k$ and last
$k$ components, respectively, of $d=G^T P^T w$.
Now we must consider two cases. First, assume that
\smash{$\rank(\tA)=\rank(A)$},
so adding the new row to $A$ did not change its rank. This implies
that $d_1=0$, and we can apply Givens rotations $G_1,\ldots G_k$ to
the right-hand side of \eqref{eq:tc} so that
\begin{equation*}
\tR = G_k^T \ldots G_1^T \left[\begin{array}{cc}
0 & d_2^T \\ 0 & R_1 \\ 0 & 0
\end{array}\right] =
\left[\begin{array}{cc} 0 & \tR_1 \\ 0 & 0 \end{array}\right],
\end{equation*}
where \smash{$\tR_1\in\R^{k\times k}$} is upper triangular.
Letting
\begin{equation*}
\tQ = \left[\begin{array}{cc}
0 & Q_1 \\
1 & 0 \\
0 & Q_2
\end{array}\right]
G_1\ldots G_k,
\end{equation*}
we complete the desired decomposition
\smash{$\tA P G = \tQ\tR$}. Note
that this QR decomposition is of the appropriate form to compute
the minimum $\ell_2$ norm solution of the least squares problem
\eqref{eq:ls3}.

The second case to consider is \smash{$\rank(\tA)>\rank(A)$}, which
means that adding the new row to $A$ increased the rank. Then at least
one component of $d_1$ is nonzero, and we can apply Givens rotations
 $G_1,\ldots G_{n-k}$ to the right-hand side of \eqref{eq:tc} so that
\begin{equation*}
\tR =
\left[\begin{array}{cc}
d_1^T & d_2^T \\ 0 & R_1 \\ 0 & 0
\end{array}\right]
G_1 \ldots G_{n-k}
=
\left[\begin{array}{cc} 0 & \tR_1 \\ 0 & 0 \end{array}\right],
\end{equation*}
where $\tR_1\in\R^{(k+1)\times (k+1)}$ is upper triangular.
We let
\begin{equation*}
\tQ = \left[\begin{array}{cc}
0 & Q_1 \\
1 & 0 \\
0 & Q_2
\end{array}\right]
\;\;\;\text{and}\;\;\;
\tG = G G_1 \ldots G_{n-k},
\end{equation*}
and observe that \smash{$\tA P \tG = \tQ \tR$} is a QR decomposition
of the desired form, so that we may compute the minimum $\ell_2$
norm minimizer of \eqref{eq:ls3}. Finally,
in either case (an increase in rank or not), we used
$O(n)$ Givens rotations, and so this update procedure requires
$O(n\cdot\max\{m,n\})$ operations.

Alternatively, suppose that \smash{$\tA\in\R^{(m-1)\times n}$} is
formed by removing the $i$th row of $A$, so
\begin{equation*}
A = \left[\begin{array}{c}
A_1 \\ w^T \\ A_2 \end{array}\right]
\;\;\;\text{and}\;\;\;
\tA = \left[\begin{array}{c}
A_1 \\ A_2
\end{array}\right],
\end{equation*}
where $A_1 \in \R^{(i-1)\times n}$, $A_2\in\R^{(m-i)\times n}$, and
$w\in\R^n$ is the row to be deleted. We follow the same arguments as
in the full rank case: we let $q\in\R^m$ denote the $i$th row of $Q$,
and compute Givens rotations $G_1,\ldots G_{m-1}$ such that
\begin{equation*}
G_{m-1}^T \ldots G_1^T q = se_1,
\end{equation*}
with $e_1=(1,0,\ldots 0) \in\R^m$ and $s=\pm 1$.
Defining \smash{$\tQ=Q G_1\ldots G_{m-1}$}, we see that
\begin{equation*}
\tQ = \left[\begin{array}{cc}
0 & \tQ_1 \\
s & 0 \\
0 & \tQ_2
\end{array}\right],
\end{equation*}
for \smash{$\tQ_1\in\R^{(i-1)\times (m-1)}$} and
\smash{$\tQ_2\in\R^{(m-i)\times (m-1)}$},
and defining \smash{$\tR = G_{m-1}^T \ldots G_1^T R$}, we have
\begin{equation*}
\tR =
\left[\begin{array}{cc}
v_1^T & v_2^T \\
0 & \tR_1 \\
0 & 0
\end{array}\right],
\end{equation*}
where \smash{$\tR_1 \in \R^{k\times k}$} has zeros below its diagonal,
and  $v_1\in\R^{n-k}$, $v_2\in\R^k$. As \smash{$APG=QR=\tQ\tR$}, we
let
\begin{equation*}
\tQ_0 =
\left[\begin{array}{c}
\tQ_1 \\ \tQ_2
\end{array}\right]
\;\;\;\text{and}\;\;\;
\tR_0 = \left[\begin{array}{cc}
0 & \tR_1 \\
0 & 0
\end{array}\right],
\end{equation*}
and conclude that \smash{$\tA P G = \tQ_0\tR_0$}, which is almost the
desired QR decomposition.
We say almost because, if \smash{$\rank(\tA)<\rank(A)$} (removing the
$i$th row decreased the rank), then the diagonal of \smash{$\tR_1$}
will have a zero
element and hence it will not be upper triangular. If the $q$th diagonal
element is zero, then we can perform Givens rotations $J_1,\ldots J_{q-1}$
and $J_{q+1},\ldots J_k$ resulting in
\begin{equation*}
\bar{R}_0 = J_{q-1}^T \ldots J_1^T
\tR_0 J_{q+1} \ldots J_{k} =
\left[\begin{array}{cc}
0 & \bar{R}_1 \\
0 & 0
\end{array}\right],
\end{equation*}
where \smash{$\bar{R}_1 \in \R^{(k-1)\times(k-1)}$} is upper
triangular. It helps to see a picture: with its 2nd diagonal element
zero,  \smash{$\tR_0$} may look like
\begin{equation*}
\sw\sh
\left[\begin{array}{ccccc}
\wBox & \Box & \Box & \Box & \Box \\
& & & \Box & \Box \\
& & & \Box & \Box \\
& & & & \Box \\
& & & & \wBox
\end{array}\right],
\end{equation*}
so applying 2 Givens rotations to the rows,
\begin{equation*}
\sw\sh
J_2^T J_1^T \cdot
\left[\begin{array}{ccccc}
\wBox & \Box & \Box & \Box & \Box \\
& & & \Box & \Box \\
& & & \Box & \Box \\
& & & & \Box \\
& & & & \wBox
\end{array}\right] =
\left[\begin{array}{ccccc}
\wBox & \Box & \Box & \Box & \Box \\
& & & \Box & \Box \\
& & & & \Box \\
& & & & \wBox \\
& & & & \wBox
\end{array}\right].
\end{equation*}
and a single Givens rotation to the columns,
\begin{equation*}
\sw\sh
\left[\begin{array}{ccccc}
\wBox & \Box & \Box & \Box & \Box \\
& & & \Box & \Box \\
& & & & \Box \\
& & & & \wBox \\
& & & & \wBox
\end{array}\right] \cdot J_4 =
\left[\begin{array}{ccccc}
\wBox & \wBox & \Box & \Box & \Box \\
& & & \Box & \Box \\
& & & & \Box \\
& & & & \wBox \\
& & & & \wBox
\end{array}\right],
\end{equation*}
which has desired form.
Letting \smash{$\bar{Q}_0 = \tQ_0 J_1 \ldots J_{q-1}$} and
\smash{$\tG=G J_{q+1} \ldots J_k$}, we have constructed the proper
QR decomposition \smash{$\tA P \tG = \bar{Q}_0 \bar{R}_0$}. We used
$O(m)$ Givens rotations, and $O(m\cdot\max\{m,n\})$ operations in
total.

\subsection{Adding or removing a column}
\label{app:updefcol}

If \smash{$\tA$} has one more or less column than $A$, then one cannot
obviously update the QR decomposition $APG = QR$ of $A$ to
obtain such a decomposition for \smash{$\tA$}.
Note that, if \smash{$\tA$} differs from $A$ by one row, then
\smash{$\tA P G$} also differs from $APG$ by one row; but if
\smash{$\tA$} differs from $A$ by one column, then \smash{$\tA$} does
not even have the appropriate dimensions for
post-multiplication by $PG$.

However, because adding or a removing a column to $A$ is the same as adding or
removing a row to $A^T$, we can compute a QR decomposition $A^T PG = QR$,
where now $P \in \R^{m\times m}$, $G \in \R^{m\times m}$, $Q \in \R^{n\times n}$,
and $R \in \R^{n\times m}$, and update it using
the strategies disussed in the previous section. The update procedure for
addition requires $O(m\cdot\max\{m,n\})$ operations, and that for removal
requires $O(n\cdot\max\{m,n\})$ operations. Hence, to be clear, we first
compute the decomposition $A^T PG = QR$ in order to solve the initial least
squares problem \eqref{eq:ls2}, as described in Section \ref{app:mlst}, and then
update it to form
\smash{$\tA^T \tP\tG = \tQ\tR$} (for some \smash{$\tP,\tG,\tQ,\tR$}),
which we use to solve \eqref{eq:ls3}.

\section{More plots from the Chicago crime data example}
\label{app:crime}

The figures below display 5 more solutions along the fused lasso path
fit to the Chicago crime data example, corresponding to
2, 5, 10, 15, and 25 degrees of freedom.

\begin{figure}[p]
\includegraphics[width=\textwidth]{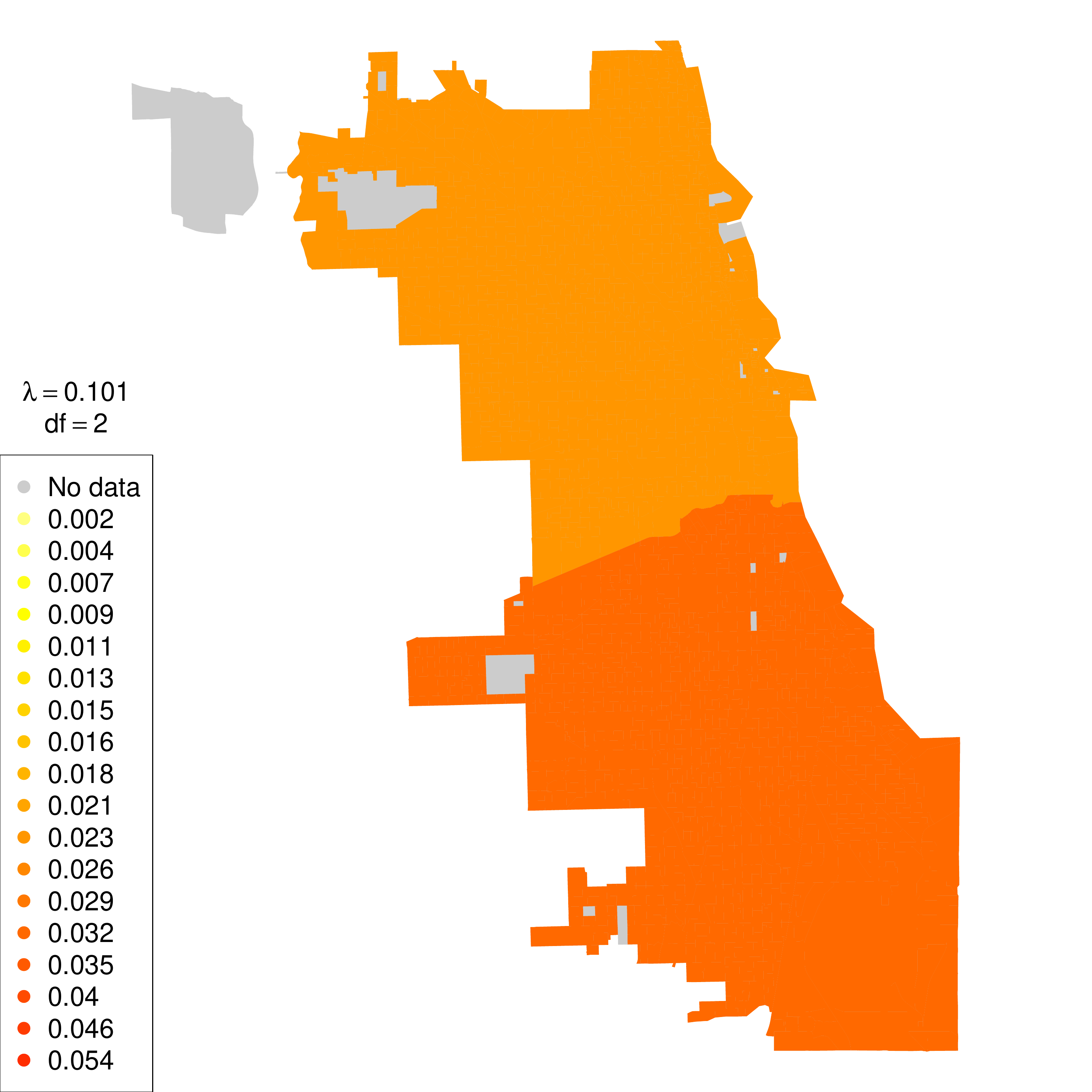}
\caption{\small\it Fused lasso solution, from the Chicago crime data 
  example, with $\lambda=0.101$.}
\end{figure}

\begin{figure}[p]
\includegraphics[width=\textwidth]{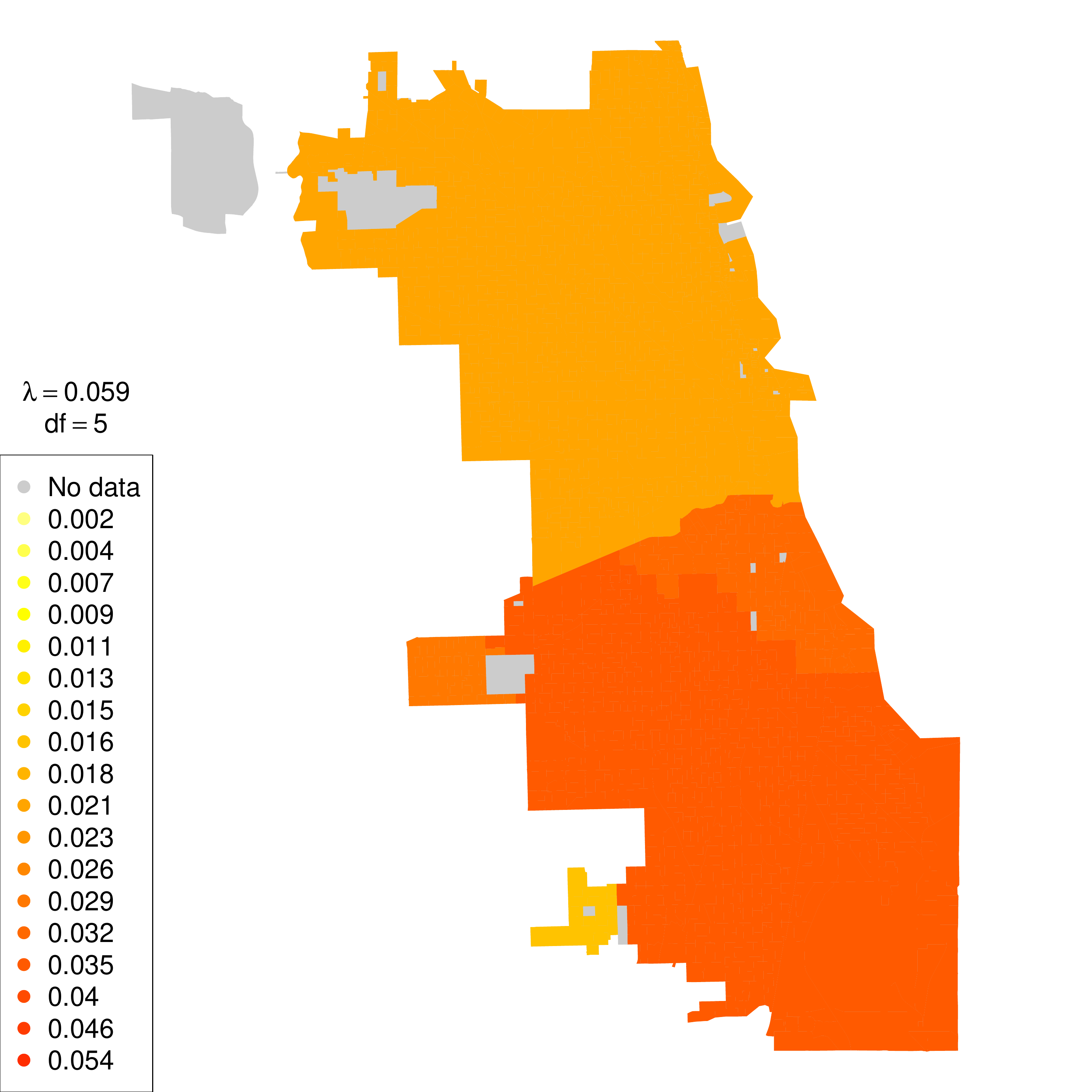}
\caption{\small\it Fused lasso solution, from the Chicago crime data 
  example, with $\lambda=0.059$.}
\end{figure}

\begin{figure}[p]
\includegraphics[width=\textwidth]{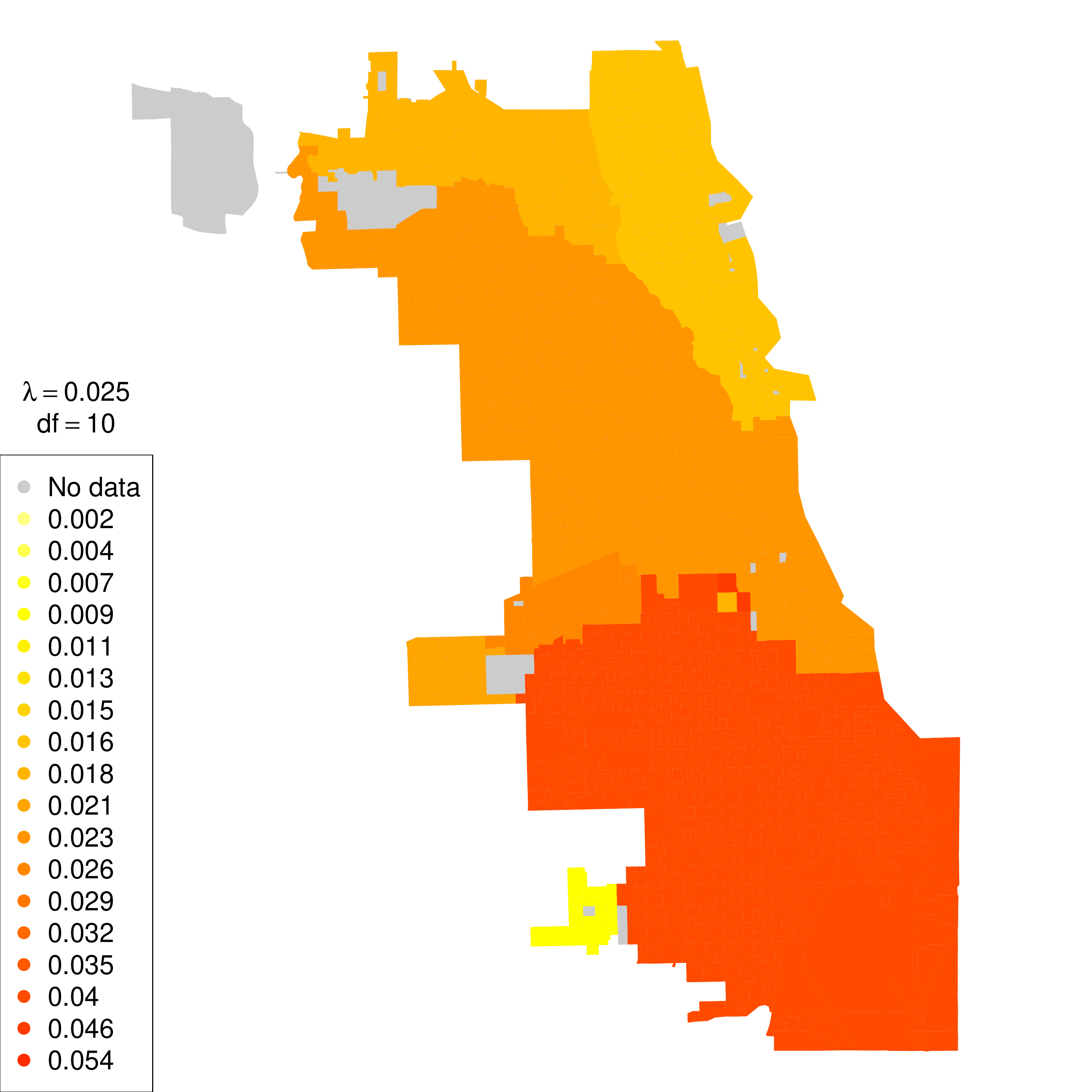}
\caption{\small\it Fused lasso solution, from the Chicago crime data 
  example, with $\lambda=0.025$.}
\end{figure}

\begin{figure}[p]
\includegraphics[width=\textwidth]{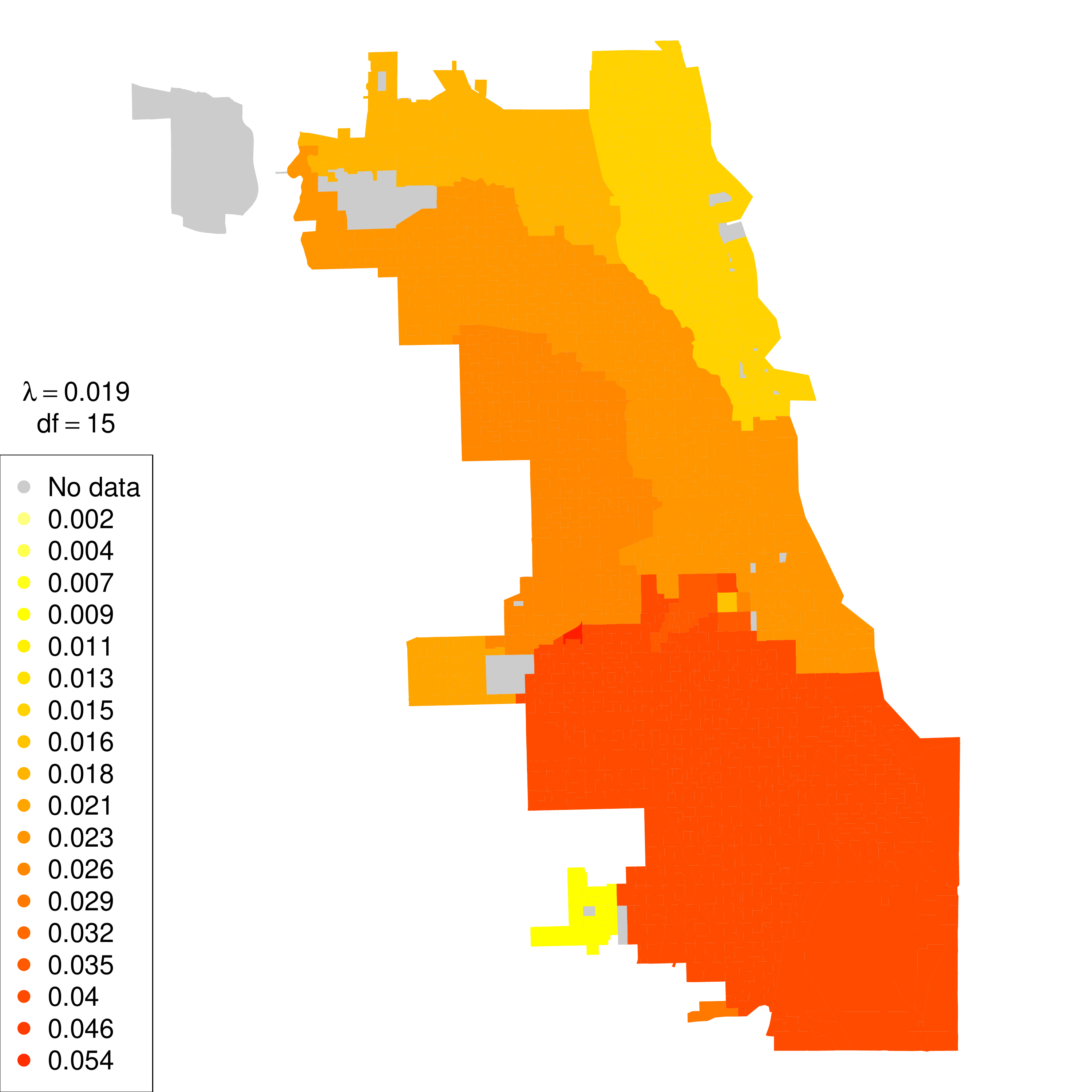}
\caption{\small\it Fused lasso solution, from the Chicago crime data 
  example, with $\lambda=0.019$.}
\end{figure}

\begin{figure}[p]
\includegraphics[width=\textwidth]{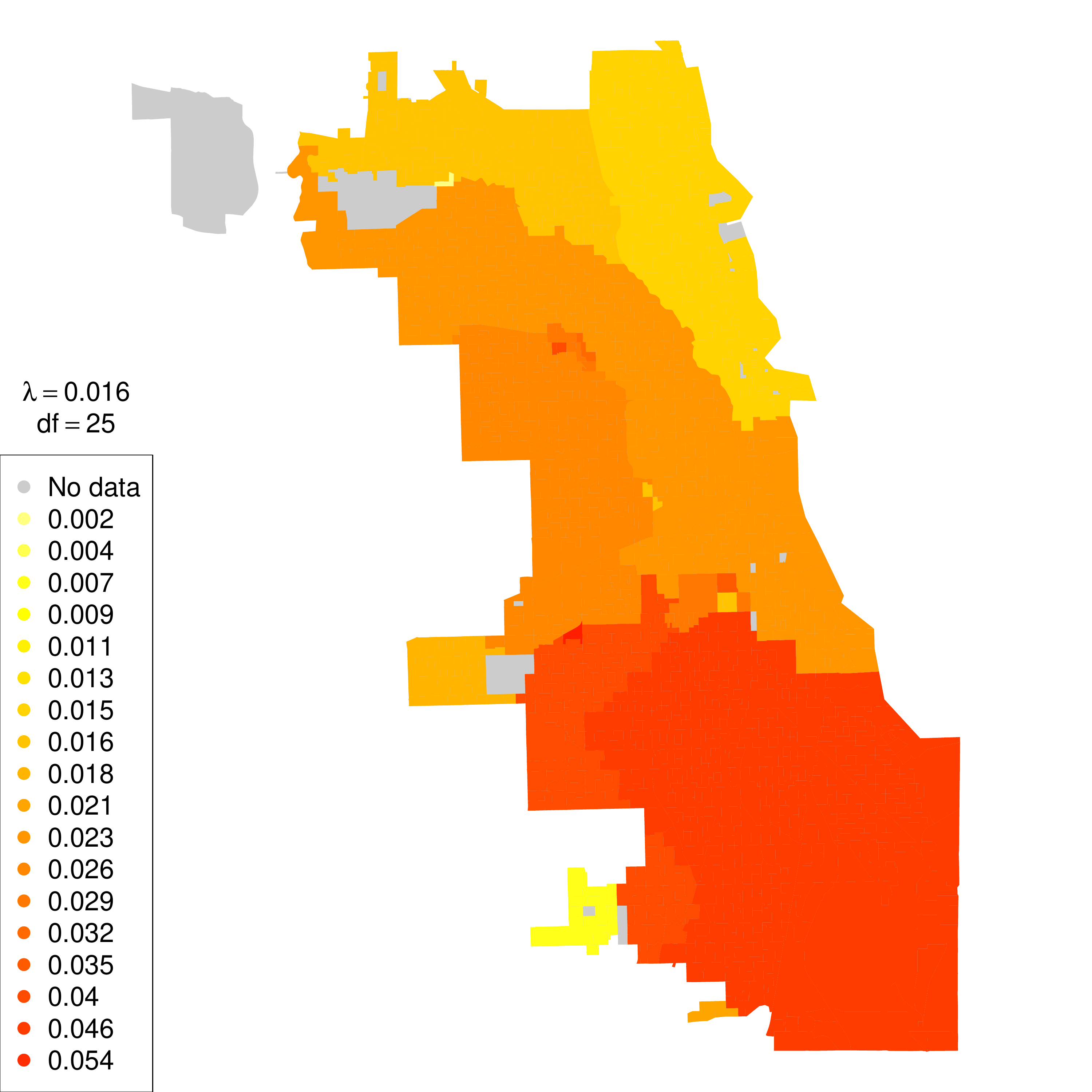}
\caption{\small\it Fused lasso solution, from the Chicago crime data 
  example, with $\lambda=0.016$.}
\end{figure}

\bibliographystyle{agsm}
\bibliography{ryantibs}

\end{document}